\documentclass[
	,paper=letter
	,abstract=true
	,fontsize=11pt
	,final
]{scrartcl}
\pdfoutput=1
\usepackage{scrhack}

\usepackage[utf8]{inputenc}

\usepackage[german, ngerman, english, canadian]{babel}
\usepackage[babel]{csquotes}		
\usepackage[ 
	,babel=true
	]{microtype}

\usepackage{amsthm,
thmtools,dsfont,thm-restate}
\usepackage{mathtools}
\allowdisplaybreaks
\declaretheorem[]{theorem}
\declaretheorem[numberlike=theorem, style=definition]{definition}
\declaretheorem[numberlike=theorem]{lemma}
\declaretheorem[numberlike=theorem]{corollary}

\usepackage{enumerate}

\usepackage[dvipsnames*,svgnames]{xcolor}

\usepackage[
	,breaklinks=true
	,colorlinks=true
	,allcolors=MidnightBlue
	,hypertexnames=false
	,pdftitle={Network Farthest-Point Diagrams and their Application to Feed-Link Network Extension}
	,pdfsubject={Research Article}
	,pdfauthor={Prosenjit Bose, Kai Dannies, Jean-Lou De Carufel, Christoph Doell, Carsten Grimm, Anil Maheshwari and Michiel Smid}
	,pdfkeywords={Network Voronoi Diagrams, Farthest Neighbor Queries, Feed-Link Problems, Data Structures, Computational Geometry}
]{hyperref}

\usepackage[
	,capitalise
	,nameinlink
	,noabbrev
	]{cleveref}
\usepackage{caption,subcaption}
\usepackage{afterpage}
\usepackage{booktabs}

\usepackage[cm, headings]{fullpage}
\usepackage{scrpage2}
\usepackage[all=normal, paragraphs=tight]{savetrees}
\newcommand{\Oh}{\ensuremath{O}}

\newcommand{\Z}{\ensuremath{\mathds{Z}}}
\newcommand{\R}{\ensuremath{\mathds{R}}}
\newcommand{\N}{\ensuremath{\mathds{N}}}

\DeclarePairedDelimiter\abs{\lvert}{\rvert}
\DeclarePairedDelimiter\set{\lbrace}{\rbrace} 

\DeclareMathOperator{\ecc}{ecc}

\DeclareMathOperator{\ED}{\mathcal{ED}}
\DeclareMathOperator{\FD}{\mathcal{FD}}

\usepackage[
	,backend=bibtex
	,style=numeric-comp
	,maxcitenames=3
	,maxbibnames=999
    ,url=false
	,doi=false,
	,isbn=false
	,eprint=true
	,firstinits=false
	]{biblatex}
\title{Network Farthest-Point Diagrams\thanks{This research has been partially funded by NSERC and FQRNT. A preliminary version of this work was presented at the 24th Canadian Conference on Computational Geometry~\cite{bose2012farthest} and was part of the Diplomarbeit (Master's thesis) of the fifth author~\cite{grimm2012charting}.}}
\subtitle{and their Application to  Feed-Link Network Extension}
\author{Prosenjit Bose\thanks{Computational Geometry Lab, School of Computer Science,
        Carleton University}%, {jit@scs.carleton.ca}}%, jdecaruf@cg.scs.carleton.ca, carsten_grimm@carleton.ca, anil@scs.carleton.ca, michiel@scs.carleton.ca}
        \and Kai Dannies\thanks{Institut für Simulation und Graphik, Fakultät für Informatik,  Otto-von-Guericke-Universität Mag\-de\-burg}%, \email{carsten.grimm@ovgu.de}}%, {\tt kai.dannies@st.ovgu.de, christoph.doell@st.ovgu.de, carsten.grimm@ovgu.de}}
        \and Jean-Lou De Carufel\footnotemark[2]
        \and Christoph Doell\footnotemark[3]
        \and Carsten Grimm\footnotemark[2]\ \footnotemark[3]
        \and Anil Maheshwari\footnotemark[2] 
        \and Stefan Schirra\footnotemark[3]
        \and Michiel Smid\footnotemark[2]}

\begin{document}%
\maketitle%
\begin{abstract}%
 Consider the continuum of points along the edges of a network, i.e., an undirected graph with positive edge weights. We measure distance between these points in terms of the shortest path distance along the network, known as the \emph{network distance}. Within this metric space, we study farthest points. 

We introduce network farthest-point diagrams, which capture how the farthest points---and the distance to them---change as we traverse the network. We preprocess a network \(G\) such that, when given a query point \(q\) on \(G\), we can quickly determine the farthest point(s) from \(q\) in \(G\) as well as the farthest distance from \(q\) in \(G\). Furthermore, we introduce a data structure supporting queries for the parts of the network that are farther away from \(q\) than some threshold \(R > 0\), where \(R\) is part of the query.

%\begin{pdfsidelinecomment}[color=Blue]{Old}
%{\color{blue} Consider the continuum of points along the edges of a network, i.e., an embedded undirected graph with positive edge weights. We measure the distance between these points as the shortest path distance along the edges of the network, known as the \emph{network distance}. We introduce two new concepts: the eccentricity diagram and the network farthest-point diagram. Network farthest-point diagrams encode the location of farthest points, whereas eccentricity diagrams capture the network distance to them. Apart from farthest points we also discuss \(R\)-far points, i.e., points that are farther away from a query point than some threshold \(R \ge 0\).  We design and analyze data structures for efficient queries for \(R\)-far points, the set of farthest points, and the network distance to farthest points (eccentricity). }
%\end{pdfsidelinecomment}

We also introduce the \emph{minimum eccentricity feed-link problem} defined as follows. Given a network \(G\) with geometric edge weights and a point \(p\) that is not on \(G\), connect \(p\) to a point \(q\) on \(G\) with a straight line segment \(pq\), called a \emph{feed-link}, such that the largest network distance from \(p\) to any point in the resulting network is minimized. We solve the minimum eccentricity feed-link problem using eccentricity diagrams. In addition, we provide a data structure for the query version, where the network \(G\) is fixed and a query consists of the point \(p\).
\end{abstract}%

\section{Introduction}%
 
We are given a network, i.e., an undirected graph with positive edge weights. We consider the continuum of points along the edges of this network and measure distance between these points in terms of the shortest path distance along the network. Within this metric space, we study farthest points. 

We introduce network farthest-point diagrams, which capture how the farthest points---and the distance to them---change as we traverse the network. We preprocess a network \(G\) such that, when given a query point \(q\) on \(G\), we can quickly determine the farthest point(s) from \(q\) in \(G\) as well as the farthest distance from \(q\) in \(G\). Furthermore, we introduce a data structure supporting queries for the parts of the network that are farther away from \(q\) than some threshold \(R > 0\), where \(R\) is part of the query.

This has applications to location analysis: Think of the network as roads in a city. When choosing the location for a new service facility, an urban engineer might want to know the parts of the network that are close-by (well-served) and the parts that are far away (ill-served). For instance, the farthest distance from a hospital influences the worst case response time of an emergency unit send from that hospital. Depending on the type of facility, we may be interested in locations with minimal farthest distance (network centers) or locations with maximal farthest distance (peripheral points) or points with farthest distance from some given location. 

We use the network farthest-point diagrams to solve the following network extension problem. Given a network \(G\) with geometric edge weights and a point \(p\) that is not on \(G\), connect \(p\) to a point \(q\) on \(G\) with a straight line segment \(pq\), called a \emph{feed-link}, such that the largest distance from \(p\) to any point in the resulting network is minimized. In terms of our example, this corresponds to the task of connecting a hospital to a network of roads minimizing the worst-case emergency unit response time. The main difficulty of feed-link problems \cite{aronov2011connect} stems from allowing feed-links to connect to any location in the network---not just to a few candidate locations---and taking every point on the network into account when measuring the utility of the feed-link.

\subsection{Related Work}

In their original work, \textcite{aronov2011connect} introduced the feed-link problem with this example, i.e., connecting a new hospital to a network of roads. They consider a target function to measure the utility of a feed-link, other than its length. \textcite{aronov2011connect} seek a feed-link that minimizes the worst-case detour one may take from any point on the network by traveling along the roads to the hospital as opposed to flying directly. The detour, or \emph{dilation}, between two points \(p\) and \(q\) on a network is measured as the ratio of the distance of \(p\) and \(q\) via the network and the length of the straight line connecting \(p\) and \(q\). We refer to this problem as the \emph{minimum dilation feed-link problem}. \Cref{fig::sketch_feedlinks} illustrates dilation and the travel time to a farthest point (eccentricity) as target functions for the feed-link problem. See, for instance, \textcite{gruene2006geometric} for a comprehensive summary of dilation and its properties.
\begin{figure}
	\centering
	\begin{subfigure}[b]{.48\linewidth}
		\centering
		\includegraphics[scale=1,page=3]{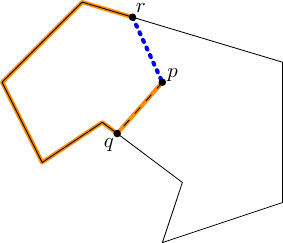}
		\caption{The dilation of \(p\) and a point \(r\) on \(G\) is the ratio between the length of a shortest path (orange) from \(r\) to \(p\) via the extended network \(G+pq\) versus the Euclidean distance of \(p\) and \(r\) (blue, dotted). \label{fig::sketch_feedlinks_dilation}}
	\end{subfigure} \quad
	\begin{subfigure}[b]{.48\linewidth}
		\centering
		\includegraphics[scale=1,page=3]{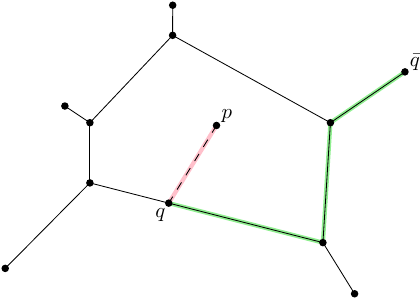}
		\caption{The eccentricity of \(p\) in the extended network \(G+pq\) is the length of the feed-link (light red) plus the eccentricity of \(q\) in \(G\), i.e., length of a shortest path (light green) from \(q\) to one of its farthest points \(\bar q\) in \(G\). \label{fig::sketch_feedlinks_eccentricity}}
	\end{subfigure}
	\caption{An illustration of the minimum dilation feed-link problem \subref{fig::sketch_feedlinks_dilation}, the minimum eccentricity feed-link problem \subref{fig::sketch_feedlinks_eccentricity}, and their target functions. In both cases we have a network \(G\) and a point \(p\) in the plane that is connected to \(G\) via the feed-link \(pq\) (dashed) resulting in an extended network \(G+pq\).   \label{fig::sketch_feedlinks}}
\end{figure}

We evaluate a feed-link with respect to all locations on a network. This means, if we use embedded graphs to model the network, all (uncountably many) points on this embedding---and not just (finitely many) vertices---count as possible farthest points. \textcite{aronov2011connect} point out that the restriction to vertices yields a related feed-link problem where we seek the optimal feed-link to minimize the detour to a new train station in a railway system. In this scenario, we are only interested in the detour of the paths from other train stations to the new one. If we model the stations as vertices, then the target function is the \emph{stretch factor}~\cite{narasimhan2007geometric} of the new station with respect to the extended railway system.

 Among other results,  \textcite{aronov2011connect} solve the minimum dilation feed-link problem for polygonal cycles in \(\Oh(\lambda_7(n)\log n)\) time, where \(\lambda_7(n)\) is the maximum length of a Davenport-Schinzel sequence~\cite{agarwal1989sharp} of order seven in \(n\) symbols. \textcite{agarwal1989sharp} show that \(\lambda_7(n)\) is \emph{almost linear in \(n\)}, more precisely \(\lambda_7(n) \le n \cdot 2^{\Oh{(\alpha(n)^{2}\log\alpha(n))}}\), where \(\alpha(n)\) is the inverse Ackermann function. Davenport-Schinzel sequences occur in the time bound, because \textcite{aronov2011connect} rely on computing certain upper envelopes \cite{agarwal2000davenport,hershberger1989finding}. \textcite{savic2012linear} improved the result for polygonal cycles to \(\Oh(n)\) with their \emph{sliding lever} algorithm. 

The subdivision of a network into parts with common farthest points is the \emph{farthest-point Voronoi diagram}~\cite{oktabe2000spatial} on the metric space formed by the uncountably many points on the network and the network distance. Previously studied network Voronoi diagrams~\cite{oktabe2000spatial,taniar2011spatial,tran2009reverse,erwig2000graph,furutaambulance,hakimi1992partition,okabe2008generalized,kolahdouzan2004voronoi} subdivide a network with respect to only finitely many reference points, e.g., depending on which reference point is closest or farthest. \textcite{oktabe2000spatial} survey various notions of Voronoi diagrams, including some for networks. We refer the reader to \textcite{okabe2008generalized} for more types of network Voronoi diagrams and for further references. 

Creating maps of the farthest points in a network relates to problems from location analysis. For instance, in the \emph{continuous absolute \(1\)-center problem}~\cite{frank1967note} we seek a point in a network with minimal network distance to its farthest points. In terms of our prototype example, this is the problem of identifying the ideal position of a new hospital on a network. \textcite{kincaid2011exploiting} and \textcite{tansel2011discrete} survey related notions and results.

\subsection{Preliminaries and Problem Definition} \label{sec::preliminaries}
A \emph{network} is a simple, weighted, finite, connected, and undirected graph \(G=(V,E)\), where \(V\) is a finite set of points in \(\R^2\), and \(E\) is a set of line segments whose endpoints are in \(V\). Each edge \(e \in E\) has a positive weight \(w_{e} > 0\). We write \(uv\) for an edge with endpoints \(u\) and \(v\). A point \(p\) on an edge \(uv\in E\) \emph{subdivides} \(uv\) into two sub-edges \(up\) and \(pv\) such that \(w_{up} = \lambda w_{uv}\) and \(w_{pv} = (1-\lambda) w_{uv}\) for some \(\lambda \in [0,1]\). %We define the weights of the resulting sub-edges \(up\) and \(pv\) according to the fraction \(\lambda\), viz., \(w_{up} \coloneqq \lambda w_{uv}\) and \(w_{pv} \coloneqq (1-\lambda) w_{uv}\). 
%As we consider possibly non-planar networks, we need to specify that a point can only be on one edge. Thus, if a point happens to lie on the proper intersection of two edges, each intersecting edge has a copy of the point.

Consider the weighted shortest path distance \(d_G \colon V \times V \to [0,\infty)\) between vertices of \(G\) with respect to the edge weights \(w_{e}\), \(e \in E\). This can be extended to arbitrary points \(p\) and \(q\) on \(G\) by considering them to be vertices for the sake of evaluating \(d_G(p,q)\). More precisely, if we subdivide the edges containing \(p\) and \(q\) according to the above, then  \(d_G(p,q)\) is the weighted shortest path distance of \(p\) and \(q\) in the subdivided network. An example is shown in \cref{fig::network_distance_example}. We refer to this distance as the \emph{network distance} \cite{aronov2011connect,gruene2006geometric} on \(G\). 
\begin{figure}[ht]
\centering
\begin{subfigure}[t]{.47\linewidth}
\centering\includegraphics[scale=1, page=1]{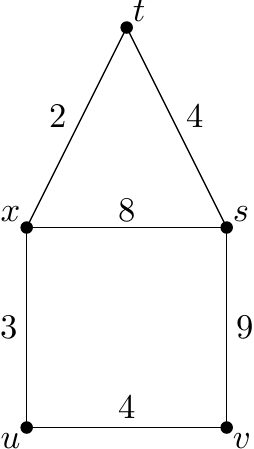}
\subcaption{A network \(G\).\label{fig::network_distance_example_input}}
\end{subfigure}\hspace{0.02\linewidth}
\begin{subfigure}[t]{.47\linewidth}
\centering\includegraphics[scale=1, page=4]{network_distance_example}
\subcaption{Subdivision of \(G\) with respect to \(p\) and \(q\).\label{fig::network_distance_example_distance}}
\end{subfigure} 
\caption{A network \(G\) \protect\subref{fig::network_distance_example_input} with edge weights as indicated. The subdivision of \(G\) according to the positions of \(p = \frac{1}{4}u + \frac{3}{4}v\) and \(q = \frac{1}{2}s + \frac{1}{2}t \) \protect\subref{fig::network_distance_example_distance}. We have \(d_G(p,q)= 10\) achieved on the highlighted path (blue).}\label{fig::network_distance_example}
\end{figure}

\begin{lemma} The network distance \(d_G(\cdot,\cdot)\) is a metric on \(G\), i.e., we have 
\begin{align}
   	d_G(p,q) &\ge 0 \text{ for all } p,q \in G,  \tag{non-negativity} \\
        d_G(p,q) &= 0 \text{ if and only if } p = q, \tag{identity of indiscernibles} \\
        d_G(p,q) &= d_G(q,p) \text{ for all } p,q \in G, \tag{symmetry} \\
        d_G(p,q) &\le d_G(p,r) + d_G(r,q) \text{ for all } p,q,r \in G\enspace . \tag{triangle inequality}
\end{align} 
\end{lemma}
\begin{proof} Since all edge weights are positive, all paths on \(G\) have non-negative weight and only paths consisting of a single point have weight zero. Thus, non-negativity and the identity of indiscernibles hold, i.e., we have \(d_G(p,q) \ge 0\) for all \(p, q \in G\) where \(d_G(p,q) = 0\) if and only if \(p = q\) . Since the edges are undirected, any path from \(p\) to \(q\) is also a path from \(q\) to \(p\). Thus, we have symmetry, i.e., \(d_G(p,q) = d_G(q,p)\). Since \(d_G(p,q)\) is the length of a shortest path from \(p\) to \(q\), there cannot be a point \(r\) on \(G\) such that \(d_G(p,q) > d_G(p,r) + d_G(r,q)\). Otherwise, the concatenation of a shortest path from \(p\) to \(r\) with a shortest path from \(r\) to \(q\) yields a path from \(p\) to \(q\) that is shorter than the shortest path from \(p\) to \(q\). Thus, the triangle inequality holds. 
\end{proof}

Let \(S \subset G\) be a subset of the points on \(G\). The \emph{boundary} of \(S\), denoted by \(\partial S\), consists of all points \(p \in S\) such that for each value \(\epsilon > 0\) there exists one point on \(G\) within network distance \(\epsilon\) from \(p\) that is in \(S\) and there exists another point on \(G\)  within network distance \(\epsilon\) from \(p\) that is not in \(S\). For example, the boundary of the set of points on the highlighted path in \cref{fig::network_distance_example_distance} consists of the points \(p\), \(q\), and \(x\).

The following definition generalizes \emph{graph eccentricity}~\cite[pp.~35--36]{harary1969graph}, which is usually introduced with respect to distances between vertices. Refer to \cref{fig::maltese_cross::example_eccentricity} for an illustration of this definition.
 \begin{definition}[Eccentricity] \label{def:eccentricity}
    Let \(G\) be a network and let \(p\) be a point on \(G\). The largest network distance towards \(p\) is called the \emph{eccentricity} of \(p\) with respect to \(G\) and it is denoted by \(\ecc_G(p)\), i.e.,
\begin{align*}
	\ecc_G(p) \coloneqq \max_{q \in G} d_G(p,q).
\end{align*}
A point \(q\) on \(G\) is called \emph{eccentric} to \(p\) if it is farthest from \(p\) with respect to the network distance on \(G\), i.e., if \(d_G(p,q) = \ecc_G(p)\). We say that a point \(q\) is \emph{eccentric} if there is a point \(p\) on \(G\) such that \(q\) is eccentric to \(p\).
\end{definition}%
\begin{figure}[hb]
\centering
\begin{subfigure}[t]{.47\linewidth}
\centering\includegraphics[scale=1, page=1]{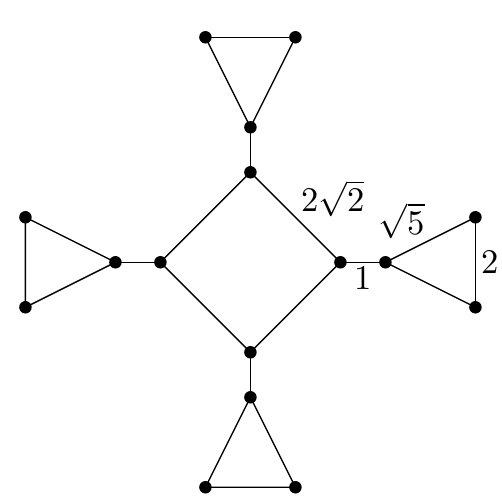}
\subcaption{A network \(G\).\label{fig::maltese_cross::weights}}
\end{subfigure}\hspace{0.02\linewidth}
\begin{subfigure}[t]{.47\linewidth}
\centering\includegraphics[scale=1, page=2]{maltese_cross_example}
\subcaption{A point \(p\) on \(G\) and a point \(\bar p\) eccentric to \(p\).\label{fig::maltese_cross::eccentricity}}
\end{subfigure} 
\caption{A network \(G\) \subref{fig::maltese_cross::weights} with (geometric) edge weights as indicated. A point \(p = \frac{1}{2}u + \frac{1}{2}v\) with its eccentric point \(\bar p = \frac{1}{2}s + \frac{1}{2}t\) is shown \subref{fig::maltese_cross::eccentricity}. The eccentricity of \(p\) is \(\ecc(p) = \frac{3}{2}\sqrt{5} + 4\sqrt{2} + 3\), achieved on the highlighted path (orange). Neither \(p\) nor \(\bar p\) are vertices of \(G\) in this example.}\label{fig::maltese_cross::example_eccentricity}
\end{figure}

In the remainder of this article we will omit the subscript indicating the underlying network \(G\) in all of the above notation when it is clear from the context. Unless stated otherwise, we refer to the network distance whenever we describe distance between points on a network. %

Our terminology for ill-serviced parts of a network is defined as follows. Let \(p\) be a point on a network \(G\), and let \(R \ge 0\). We say a point \(q\) on  \(G\) is  \emph{\(R\)-far} from \(p\), if \(d(p,q) \ge R\). Likewise, we call the set of points on \(G\) that are \(R\)-far from \(p\) the \(R\)-far sub-network of \(G\) from \(p\).  

Alongside with our goal to characterize farthest-point information and changes therein, we aim to design efficient data structures to answer the following types of queries for any point \(p\) on a given network \(G\).%
\begin{enumerate}
    \item What is the eccentricity of \(p\)?
    \item Which points on \(G\) are farthest from \(p\) with respect to the network distance?
    \item What is the \(R\)-far sub-network of \(G\) from \(p\) for some value \(R \ge 0\)? %Let \(uv\) be an edge such that \(p \in uv\). Which points on \(uv\) have the same farthest points as \(p\)?
\end{enumerate}%
We represent the query point \(p\) as a pair made of an edge \(uv\) and a value \(\lambda \in [0,1]\) with \(p \in uv\) and \(p = (1-\lambda) u + \lambda v\). We refer to a query for the eccentricity as an \emph{eccentricity query}, and we refer to a query for the set of farthest points as \emph{farthest-point-set query}. 

With the notation established above, we can formally define the term feed-link and the minimum eccentricity feed-link problem as follows. An example is shown in \cref{fig::maltese_cross::minimum_feed_link_problem}. 
\begin{figure}[!hb]
	\centering
	\begin{subfigure}[t]{.47\linewidth}
		\centering
		\includegraphics[scale=1, page=3]{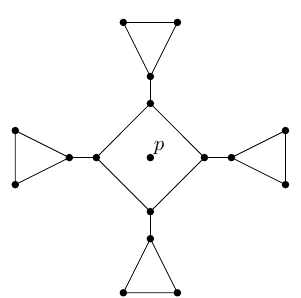}
		\caption{A network \(G\) with a point \(p \notin G\).}\label{fig::maltese_cross::feedlink_example_input}
	\end{subfigure}\hspace{0.02\linewidth}
	\begin{subfigure}[t]{.47\linewidth}
		\centering
		\includegraphics[scale=1, page=5]{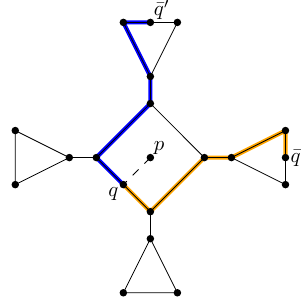}
		\caption{An optimal feed-link.}\label{fig::maltese_cross::optimal_feed_link}
	\end{subfigure} 
	\caption{An instance of the minimum eccentricity feed-link problem \subref{fig::maltese_cross::feedlink_example_input}. The anchor \(q\) of an optimal feed-link (dashed) \subref{fig::maltese_cross::optimal_feed_link} balances the distance to the eccentric points \(\bar q\) and \(\bar q'\).}\label{fig::maltese_cross::minimum_feed_link_problem}
\end{figure}
\begin{restatable}[Feed-links and the Minimum Eccentricity Feed-Link Problem]{problem}{mineccfeedlinkprobdef} \label{def::minimum_eccentricity_feed_link_problem}
  Let \(p\in \R^2\) be a point and \(G = (V,E)\) be a straight-line embedded network with geometric edge weights.%, i.e., where the weight of each edge is its Euclidean length.  	
	%We are given a network \(G = (V,E)\) with \emph{geometric edge weights}, i.e., with \(w_e = \abs{e}\) for all \(e\in E\). Further, we are given a point \(p \in \R^2\). 
	\begin{enumerate}[(i)]
	    \item A straight-line segment \(pq\) with \(q\in G\) is called a \emph{feed-link} connecting \(p\) to \(G\) and \(q\) is called the \emph{anchor} of this feed-link. The network that results from subdividing the edge containing \(q\) at \(q\) and adding the edge \(pq\) is denoted by \(G+pq\). It is referred to as the \emph{extension} of \(G\) by the feed-link \(pq\).
	    \item We call the task of finding a point \(q\) on \(G\) such that the feed-link \(pq\) with anchor \(q\) minimizes the eccentricity of \(p\) with respect to \(G+pq\), i.e.,  the task of finding a point \(q\) on \(G\) that minimizes the expression
		\begin{align}
			\ecc_{G+pq}(p) = \abs{pq} + \ecc_G(q) = \abs{pq} + \max_{r\in G} d_G(q,r), \label{eq::target_function_minimum_eccentricity_feed_link_problem}
		\end{align}
		the \emph{minimum eccentricity feed-link problem}. 	
	\end{enumerate} 
\end{restatable}%

%In the query version of the minimum eccentricity problem, we have a fixed network \(G\). A query consists of the point \(p\) that needs to be connected to \(G\) and which is unknown a-priori. We refer to a query for an optimal feed-link as a \emph{feed-link query}. 

\subsection{Structure and Results}

The farthest-point Voronoi diagram of a set \(S\) of \(n\) points in \(\R^2\) subdivides the plane into regions with a common farthest point among those in \(S\). Given some point \(p \in \R^2\), this diagram can be used to determine the farthest point from \(p\) in \(S\) . We wish to subdivide a network in a similar fashion in order to compute the farthest points from any point \(p\) on a network. However, there are differences between the situation in the plane and on a network. For instance, when given a network, the locations of the farthest points are unknown, as shown in \cref{fig::stationary_vs_non_stationary_farthest_points}.   
\begin{figure}[!ht]
	\centering
	\begin{subfigure}[b]{.49\linewidth}
		\centering
		\includegraphics[scale=1]{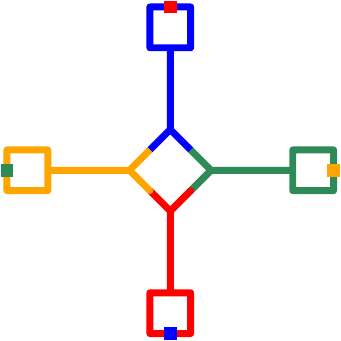}
		\caption{Stationary farthest points.\label{fig::stationary_farthest_points}}% \label{fig::example_quadlet_ecc_diag}
	\end{subfigure} 
	\begin{subfigure}[b]{.49\linewidth} 
		\centering
		\includegraphics[scale=1]{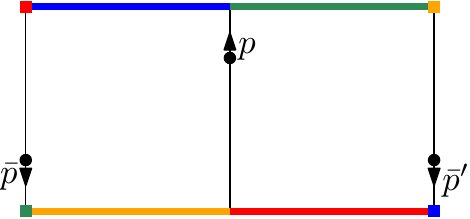}
		\caption{Moving farthest points.\label{fig::non_stationary_farthest_points}} %\label{fig::example_twocycle_cfpnvd}
	\end{subfigure}
	\caption{Two networks with different farthest point behaviour. In the network in \subref{fig::stationary_farthest_points}, each point has one of the four squares as its farthest point. In the network in \subref{fig::non_stationary_farthest_points}, each point \(p\) on the vertical middle edge has two farthest points \(\bar p\) and \(\bar p'\). Moving \(p\) upwards causes these points to move downwards. %ence, no two points on the middle edge have the same set of farthest points.   
	 \label{fig::stationary_vs_non_stationary_farthest_points}}
\end{figure}

Overcoming these issues and defining an analogue to the farthest-point Voronoi diagram in networks is one of the contributions of this work. We introduce the new notions of \emph{eccentricity diagrams} in \cref{sec::eccentricity_diagrams} and \emph{network farthest-point diagrams} in \cref{sec::network_farthest_point_diagrams}. Network farthest-point diagrams encode the location of farthest points, whereas eccentricity diagrams capture the network distance to them. We show that, contrary to intuition, these diagrams may have non-linear size. In \cref{sec::far_point_queries}, we design and analyze a data structure for efficient eccentricity, \(R\)-far, and farthest-point-set queries in networks. 

We solve the static and query version of the minimum eccentricity feed-link problem in \cref{sec::feed-links} using eccentricity diagrams. We rephrase the query version as a point location problem in a certain Voronoi diagram whose sites are the sub-edges of the network with minimal eccentricity. Solving the static version of the feed-link problem takes \(\Oh(m^2 \log n)\) time and \(\Oh(m^2)\) work space for a network with \(n\) vertices and \(m\) edges.  \Cref{tab::results} summarizes the asymptotic bounds of all the other results.
\begin{table}[!hb]
    \centering \small
    \begin{tabular}{lccc}
    \toprule
      & \textbf{Query Time} & \textbf{Pre-Processing Time} & \textbf{Space} \\ \midrule
     Eccentricity Query & \(\Oh(\log n)\) & \(\Oh(m^2 \log n)\) & \(\Oh(m^2)\)  \\
     \(R\)-Far Query & \(\Oh(k + \log n)\) & \(\Oh(m^2 \log n)\) & \(\Oh(m^2 \log n)\) \\
     Farthest-Point-Set Query & \(\Oh(k'+\log n)\) & \(\Oh(m^2 \log n)\) & \(\Oh(m^2 \log n)\) \\ \midrule 
     Feed-Link Query &  \(\Oh(\ell)\) &--- & \(\Oh(\ell)\)  \\ 
     Feed-Link Query & \(\Oh(\log \ell)\) & \(\Oh( \ell^{2+\epsilon})\)  & \(\Oh(\ell^2)\) \\
     Feed-Link Query & \(\Oh(\log \ell)\) & expected \(\Oh(\ell^2 \log \ell)\) & \(\Oh(\ell^2)\) \\ \bottomrule
    \end{tabular}
    \caption{Our results for a network with \(n\) vertices, \(m\) edges. For \(R\)-far queries, \(k\) is the number of edges containing \(R\)-far points, and for farthest-point-set queries, \(k'\) is the number of edges containing farthest points. For the results regarding feed-link queries, \(\ell \in \Oh(m^2)\) is the number of sub-edges with locally minimal eccentricity, and we assume that these sub-edges are known a-priori.\label{tab::results}}
\end{table}

\section{Eccentricity Diagrams} \label{sec::eccentricity_diagrams}%

In this section, we describe the distance to farthest points. We begin with computing the distance functions from points on an edge \(uv\) to their farthest points on an edge \(st\). Combining these functions yields the distance from the points on \(uv\) to their farthest points in the entire network. This approach leads us to a representation of the distance to farthest points (eccentricity), which we call the \emph{eccentricity diagram} of a network. From eccentricity diagrams, we derive a data structure for eccentricity queries, which we extend to a data structure for farthest-point-set queries in \cref{sec::network_farthest_point_diagrams}. 

\subsection{The Shape of the Eccentricity Function}

We use a result by \textcite{frank1967note} to compute the eccentricity of all points on a network \(G=(V,E)\). \textcite{frank1967note} seeks a point on a network with smallest distance to its farthest points, i.e., a minimum of the eccentricity. He finds this minimum by determining a point with minimum eccentricity on each edge \(uv\) and then picking a point with the smallest eccentricity among these candidates. \textcite{frank1967note} computes the eccentricity on an edge as follows. Let \(uv\) and \(st\) be edges of \(G\).  We define \(\phi_{uv}^{st} \colon [0,1] \to [0,\infty)\) as the mapping from \(\lambda \in [0,1]\) to the largest network distance from the point \(p\) on \(uv\) with \(p = p(\lambda) = (1-\lambda)u + \lambda v\) to any point \(q\) on edge \(st\), i.e., %  
\begin{align}
	\phi_{uv}^{st}(\lambda) \coloneqq \max_{q \in st} d(p(\lambda), q) = \max_{q \in st} d((1-\lambda)u + \lambda v, q) \enspace . \label{eq::definition::phi_uv_st}
\end{align}%
The upper envelope of the functions \(\phi_{uv}^{st} \), over all edges \(st \in E\), is the eccentricity of \(p(\lambda)\) on \(uv\), since%
\begin{align*}
   \ecc(p(\lambda)) = \max_{q \in G} d(p(\lambda),q)  = \max_{st \in E} \max_{q \in st} d(p(\lambda),q) = \max_{st \in E} \phi_{uv}^{st}(\lambda) \enspace .
\end{align*}%
We begin the analysis of the functions \(\phi_{uv}^{st}\) and their upper envelope with an auxiliary lemma.

\begin{lemma} \label{thm::point_to_edge_network_distance}
Let \(G\) be a network, let \(ab\) be an edge of \(G\), and let \(x\) be a point on \(G\) that is not in the interior of the edge \(ab\). The network distance from \(x\) to the farthest point \(\bar x\) from \(x\) among all points on \(ab\) is 
\begin{align}
	d(x,\bar x) = \max_{y \in ab} d(x,y) = \frac{d(x,a) + w_{ab} + d(x,b)}{2}. \label{eq::point_to_edge_network_distance}
\end{align}
\end{lemma}
\begin{proof} Let \(\bar x\) be the farthest point from \(x\) on \(ab\). Then the shortest path from \(x\) to \(\bar x\) via \(a\) has the same length as the one via \(b\), i.e.,
\(
	d(x,\bar x) = d(x,a) + w_{a\bar x}=d(x,b) + w_{b\bar x}
\). This yields
\begin{align*}
	d(x,\bar x) = \frac{2d(x,\bar x)}{2} = \frac{d(x,a) + w_{a\bar x} + d(x,b) + w_{\bar xb}}2, 
\end{align*}
and the result follows as \(w_{ab} = w_{a\bar x} + w_{\bar xb}\).
\end{proof}

The following \lcnamecref{thm::edge_to_edge_network_distance_function} by \textcite{frank1967note} describes the function \(\phi_{uv}^{st}\) for two distinct edges \(uv\) and \(st\).  Refer to \cref{fig::edge_to_edge_distance::illustration} for an illustration of the notation and the result.%
\begin{figure}[ht]
	\centering
	\begin{subfigure}[b]{0.49\linewidth}
		\centering
		\includegraphics[page=1]{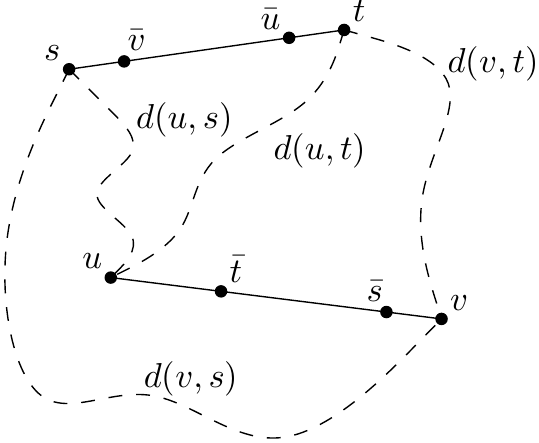}
		\subcaption{Two edges in a network.} \label{fig::edge_to_edge_distance::subdivision}
	\end{subfigure}%
	\begin{subfigure}[b]{0.49\linewidth}
		\centering 
		\includegraphics{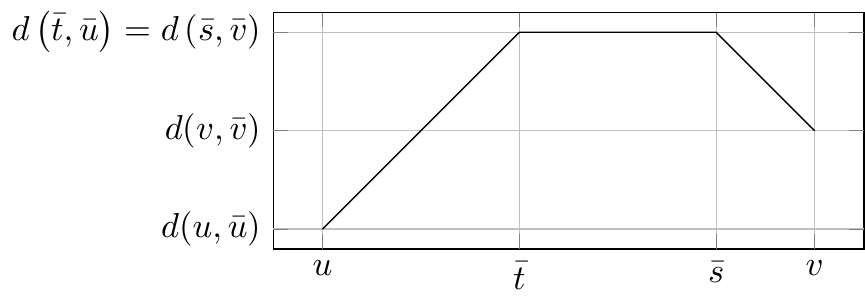}
		\subcaption{The function \(p(\lambda) \mapsto \phi_{uv}^{st}(\lambda) = \max_{q \in st}d(p(\lambda),q)\), where \(p(\lambda) = (1- \lambda) u + \lambda v\) and \(\lambda \in [0,1]\).} \label{fig::plot::edge_to_edge_distance}
	\end{subfigure}
	\caption{An illustration of \cref{thm::edge_to_edge_network_distance_function} and its notation for distinct edges \(uv\) and \(st\).   On the left \subref{fig::edge_to_edge_distance::subdivision}, the edges are subdivided with respect to \(\bar s\), the farthest point from \(s\) on \(uv\); to \(\bar t\), the farthest point from \(t\) on \(uv\); to \(\bar u\), the farthest point from \(u\) on \(st\); and to \(\bar v\), the farthest point from \(v\) on \(st\). The positions of \(\bar s\), \(\bar t\), \(\bar u\), and \(\bar v\) follow from \cref{thm::point_to_edge_network_distance}.  On the right \subref{fig::plot::edge_to_edge_distance}, we see a plot of the network distance from points \(p\) on \(uv\) to their farthest point on \(st\), according to \cref{thm::edge_to_edge_network_distance_function}. The slopes of the ascending and descending segments are equal up to their sign.\label{fig::edge_to_edge_distance::illustration}}
\end{figure}%
\begin{lemma}[\textcite{frank1967note}] \label{thm::edge_to_edge_network_distance_function}
Let \(uv\) and \(st\) be two distinct edges of a network \(G\).  Let \(\bar u\)  (respectively \(\bar v\)) be the farthest point from \(u\) (respectively from \(v\)) on \(st\). Likewise, let \(\bar t\)  (respectively \(\bar s\)) be the farthest point from \(t\) (respectively from \(s\)) on \(uv\). Without loss of generality, we have \(\bar t \in u\bar s\) (otherwise swap \(s\) and \(t\)). Then we have
	\begin{align}
			\phi_{uv}^{st}(\lambda) =
			\begin{dcases}
				\lambda w_{uv} + d(u,\bar u) &\text{, if } p(\lambda) \in u\bar t\\
				d(\bar t,\bar u) &\text{, if } p(\lambda) \in \bar t\bar s \\
				(1-\lambda)w_{uv}  + d(v,\bar v)&\text{, if } p(\lambda) \in \bar sv
			\end{dcases}, \label{eq::phi_uv_st}
		\end{align}
	where \(p(\lambda) = (1-\lambda) u + \lambda v\) and \(\lambda \in [0,1]\).
%\begin{enumerate}[(i)]
%	\item For all points \(p\) on \(u\bar t\), the point \(\bar u\) is farthest among all points on \(st\) with
%	\begin{align*}
%		d(p,\bar p) &= d(p,\bar u) = w_{up} + d(u,\bar u) = w_{up} + \frac{w_{st} + d(u,s) + d(u,t)}{2}.
%	\end{align*}
%	\item As \(p\) moves from \(\bar t\) to \(\bar s\), its farthest point \(\bar p\) on \(st\) moves from \(\bar u\) to \(\bar v\) and stays at distance \(d(\bar t, \bar u)\): 
%	
%	If \(p \in \bar t\bar s\) with \(p = (1-\mu) \bar t + \mu \bar s\) for some \(\mu \in [0,1]\) then  \(\bar p = (1-\mu)\bar u + \mu \bar v \) and
%	\begin{align*}
%		d(p,\bar p) &= d(\bar t,\bar u) = \frac{w_{uv} + w_{st} + d(u,s) + d(v,t))}{2} = d(\bar s, \bar v).
%	\end{align*}
%		\item For all points \(p\) on \(\bar sv\), the point \(\bar v\) is farthest among all points on \(st\) with
%	\begin{align*}
%		d(p,\bar p) &= d(p,\bar v) = w_{pv} + d(v,\bar v) = w_{pv} + \frac{w_{st} + d(v,s) + d(v,t)}{2}.
%	\end{align*}
%	\end{enumerate}
\end{lemma}%
The proof of \cref{thm::edge_to_edge_network_distance_function} was omitted by \textcite{frank1967note}. We add it for the sake of completeness.%
\begin{proof} The three cases of \cref{thm::edge_to_edge_network_distance_function} are illustrated in \cref{fig::edge_to_edge_distance::case_distinction}. Case~(3) is symmetric to Case~(1). 
\begin{figure}[ht]
	\begin{subfigure}[t]{.33\linewidth}
		\centering 
		\includegraphics[page=3]{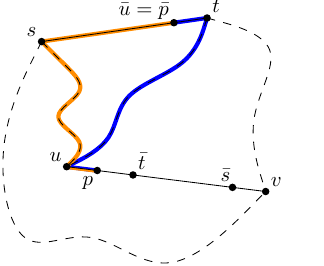}
		\subcaption{Case~(1): \(p \in u\bar t\).}\label{fig::edge_to_edge_distance::p_on_u_bar_t}
	\end{subfigure}%
	\begin{subfigure}[t]{.33\linewidth}
		\centering
		\includegraphics[page=4]{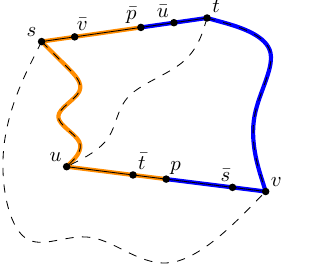}
		\subcaption{Case~(2): \(p \in \bar t\bar s\).}\label{fig::edge_to_edge_distance::p_on_bar_t_bar_s}
	\end{subfigure}
	\begin{subfigure}[t]{.33\linewidth}
		\centering
		\includegraphics[page=5]{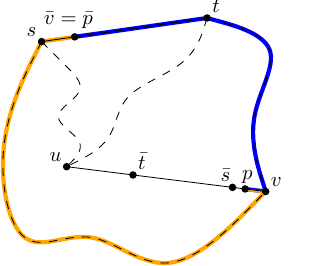}
		\subcaption{Case~(3): \(p \in \bar sv\).}\label{fig::edge_to_edge_distance::p_on_bar_s_v}
	\end{subfigure}
	\caption{An illustration of the three cases in \cref{thm::edge_to_edge_network_distance_function} from left to right. The paths from \(p\) to \(\bar p\) entering the edge \(st\) via \(s\) (orange) and via \(t\) (blue) are highlighted.\label{fig::edge_to_edge_distance::case_distinction}}
\end{figure}%

Let \(\bar p\) be the farthest point from \(p = p(\lambda)\) on \(st\). 
\begin{enumerate}[C{a}se (1):]
	\item Let \(p\) be located on \(u\bar{t}\) as shown in \cref{fig::edge_to_edge_distance::p_on_u_bar_t}. 
	
	We show that \(\bar u\) is the farthest point from \(p\) on \(st\), i.e., we show \(\bar p = \bar u\). Since \(p \in u\bar t\), there is a  shortest path from \(p\) to \(t\) that includes the sub-edge \(up\). Therefore, we have \(d(p,t) = w_{up} + d(u,t)\). Likewise, we have \(d(p,s) = w_{up} + d(u,s)\), since \(p\) is on \(u\bar s\), as well. With \cref{thm::point_to_edge_network_distance} we obtain
\begin{align*}
 d(p,\bar p) &= \frac{d(p,s) + d(p,t) + w_{st}}2,  \\
	&= \frac{w_{up} + d(u,s) + w_{up} + d(u,t) + w_{st}}2, \\ 
	&= w_{up} + \frac{d(u,s) + d(u,t) + w_{st}}2, \\ 
	&= w_{up} + d(u,\bar u).
\end{align*}

\item  Let \(p\) be located on \(\bar{t}\bar{s}\) as shown in \cref{fig::edge_to_edge_distance::p_on_bar_t_bar_s}.

    We show \(d(p,\bar p) = d(\bar u, \bar t) = d(\bar s, \bar v)\), which holds because \(\bar p\) moves from \(\bar u\) to \(\bar v\) as \(p\) moves from \(\bar t \) to \(\bar s\).
    
	Since \(p \in \bar tv\) and \(p \in u\bar s\), there is a shortest path from \(p\) to \(t\) that includes \(pv\) and there is a shortest path from \(p\) to \(s\) that includes \(pu\). Therefore, we have \(d(p,t) = w_{pv} + d(v,t)\) and \(d(p,s) = w_{up} + d(u,s)\). Plugging this into \eqref{eq::point_to_edge_network_distance} from \cref{thm::point_to_edge_network_distance} yields
	\begin{align*}
  d(p,\bar p) &= \frac{d(p,s) + d(p,t) + w_{st}}2,  \\
	&=  \frac{w_{up} + d(u,s) + w_{pv} + d(v,t) + w_{st}}2, \\ 
	&= \frac{w_{uv} + d(u,s) + d(v,t) + w_{st}}2.
\end{align*}%
We know from Cases~(1) and~(3) that \(\bar u\) is farthest from \(\bar t\) on \(st\), and that \(\bar v\) is farthest from \(\bar s\) on \(st\). As the above applies to the cases \(p=\bar t\) and \(p=\bar s\), we obtain \(d(p,\bar p) = d(\bar t,\bar u) = d(\bar s, \bar v)\). 

%Consider a shortest path from \(\bar t\) to \(\bar u\) that contains the sub-edge \(\bar tv\). It has the same length as a shortest path from \(p\) to \(\bar p\) that uses the sub-edge \(pv\). Both of these paths enter the edge \(st\) via \(t\), as they contain \(\bar s\). Decomposing and comparing these paths yields
%\begin{align*}
%	w_{\bar tv} + d(v,t) + w_{t\bar u} = d(\bar t,\bar u) = d(p,\bar p) = w_{pv} + d(v,t) + w_{t\bar p},
%\end{align*}
%which implies \(w_{\bar tp} = w_{\bar tv} - w_{pv} = w_{t\bar p} - w_{t\bar u} = w_{p\bar u}\), as \(p \in \bar tv\). As this holds for \(p = \bar s\) as well, we obtain \(w_{\bar t\bar s} = w_{\bar v \bar u}\) with the fact that \(\bar v\) is farthest from \(\bar s\) on \(st\). 

%Let the values \(\mu,\mu' \in [0,1]\) be such that \(p = (1-\mu) \bar t + \mu \bar s\) and \(\bar p = (1-\mu') \bar u + \mu' \bar v\). Then we have \(\mu = \frac{w_{\bar tp}}{w_{\bar t\bar s}} \) and \(\mu' = \frac{w_{\bar up}}{w_{\bar u\bar v}}\). Thus, the equations \(w_{\bar tp} = w_{p\bar u}\) and \(w_{\bar t\bar s} = w_{\bar v \bar u}\) imply \(\mu = \mu'\).
\end{enumerate}%

Summarizing the three cases yields
\begin{align*}
    \phi_{uv}^{st}(\lambda) = \max_{q \in st} d(p(\lambda), q) = 
        \begin{dcases}
			w_{up(\lambda)} + d(u,\bar u) &\text{, if } p(\lambda) \in u\bar t\\
			d(\bar t,\bar u) &\text{, if } p(\lambda) \in \bar t\bar s \\
			w_{p(\lambda)v}  + d(v,\bar v)&\text{, if } p(\lambda) \in \bar sv
	    \end{dcases},
\end{align*} 
where \(p(\lambda) = (1- \lambda) u + \lambda v\), which implies the  claim, since \( w_{up(\lambda)} = \lambda w_{uv}\) and \(w_{p(\lambda)v} = (1-\lambda) w_{uv}\). 
\end{proof}

In the next \lcnamecref{thm::edge_to_self_network_distance_function} we describe the distance from a point \(p\) on edge \(uv\) to its farthest point on \(uv\) itself---and thus the function \(\phi_{uv}^{uv}\). An illustration of \cref{thm::edge_to_self_network_distance_function} and the notation used is shown in \cref{fig::edge_to_self_network_distance_function::illustration}.%

\begin{lemma} \label{thm::edge_to_self_network_distance_function}
	Let \(uv\) be an edge of a network \(G\).  Let \(\bar u\)  (respectively \(\bar v\)) be the farthest point from \(u\) (respectively from \(v\)) on \(uv\). Further, let \(c\) be the midpoint of \(uv\). Then we have  \(\bar u \in cv\) and \(\bar v \in uc\) with \(w_{u\bar v} = w_{v\bar u}\), and 
	\begin{align*}
			\phi_{uv}^{uv}(\lambda) =
			\begin{dcases}
				\frac{w_{uv} + d(u,v)}{2} &\text{, if } p(\lambda) \in u\bar v\\
				(1-\lambda)w_{uv} &\text{, if } p(\lambda) \in \bar vc \\
				\lambda w_{uv}  &\text{, if } p(\lambda) \in c\bar u \\
				\frac{w_{uv} + d(u,v)}{2} &\text{, if } p(\lambda) \in \bar uv\\
			\end{dcases},
		\end{align*}
		where \(p(\lambda) = (1-\lambda) u + \lambda v\) and \(\lambda \in [0,1]\).
		
%	\begin{enumerate}[(i)]
%	    \item We have. 
%	    \item As \(p\) moves from \(u\) to \(\bar v\) its farthest point \(\bar p\) on \(uv\) moves from \(\bar u\) to \(v\) and stays at distance \(\frac{w_{uv} + d(u,v)}{2}\):
%	    
%		If  \(p \in u\bar v\) with \(p = (1-\mu) u + \mu \bar v\) for some \(\mu \in [0,1]\) then \(\bar p = (1-\mu)\bar u + \mu v\) with 
%		\begin{align*}
%			d(p,\bar p) = \frac{w_{uv} + d(u,v)}{2}.
%		\end{align*}
%		\item For all points \(p\) on \(\bar vc\), the point \(v\) is the farthest point among all points on \(uv\) with \(d(p,\bar p) = w_{pv}\).
%		\item For all points \(p\) on \(c\bar u\), the point \(u\) is the farthest point among all points on \(uv\) with \(d(p,\bar p) = w_{up}\).
%		\item As \(p\) moves from \(\bar u\) to \(v\) its farthest point \(\bar p\) on \(uv\) moves from \(u\) to \(\bar v\) and stays at distance \(\frac{w_{uv} + d(u,v)}{2}\):
%		
%		If  \(p \in \bar u v\) with \(p = (1-\mu) \bar u + \mu \bar v\) for some \(\mu \in [0,1]\) then \(\bar p = (1-\mu)u + \mu \bar v\) with 
%		\begin{align*}
%			d(p,\bar p) = \frac{w_{uv} + d(u,v)}{2}.
%		\end{align*}
%		\item The function \(\phi_{uv}^{uv}(\lambda) = \max_{q \in uv}d(p(\lambda),q)\) has the form 
%		 
%	\end{enumerate}
\end{lemma}
\begin{figure}[!ht]
	\centering
	\begin{subfigure}[b]{0.47\linewidth}
		\centering
		\includegraphics[page=1]{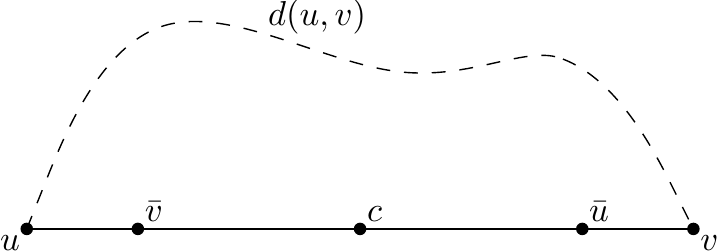}
		\subcaption{An edge in a network.\label{fig::edge_to_self_distance::subdivision}}
	\end{subfigure}\quad%
	\begin{subfigure}[b]{0.5\linewidth}
		\centering
		\includegraphics{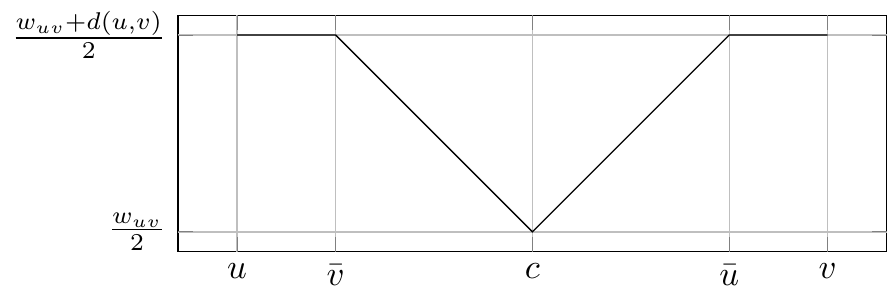}
		\subcaption{The function \(p(\lambda) \mapsto \phi_{uv}^{uv}(\lambda) =  \max_{q \in uv}d(p(\lambda),q)\), where \(p(\lambda) = (1- \lambda) u + \lambda v\) and \(\lambda \in [0,1]\).\label{fig::plot::edge_to_self_distance}}
	\end{subfigure}
	\caption{An illustration of \cref{thm::edge_to_self_network_distance_function} and its notation for an edge \(uv\). On the left \subref{fig::edge_to_self_distance::subdivision}, the edge \(uv\) is subdivided with respect to \(\bar u\), the farthest point from \(u\) on \(uv\); to \(c\), the midpoint of \(uv\); and to \(\bar v\), the farthest point from \(v\) on \(uv\). On the right \subref{fig::plot::edge_to_self_distance}, we see a plot of the network distance from points on \(uv\) to their farthest point on \(uv\) itself, according to \cref{thm::edge_to_self_network_distance_function}. The slopes of the descending and ascending segments are equal up to their sign, and the function \(\phi_{uv}^{uv}\) is symmetric with respect to  \(c\). \label{fig::edge_to_self_network_distance_function::illustration}}
\end{figure}
\begin{proof} First, we show the claims about the positions of \(\bar u\) and \(\bar v\). Using \cref{thm::point_to_edge_network_distance}, we have 
	\begin{align*}
		w_{u\bar u} = d(u,\bar u) = \frac{w_{uv} + d(u,u) + d(u,v)}{2} = \frac{w_{uv} + d(u,v)}{2} = \frac{w_{uv} + d(v,v) + d(u,v)}{2} = d(v,\bar v) = w_{v \bar v},
	\end{align*}
	which implies \(\bar u \in cv\) and \(\bar v \in uc\), since \(0 \le d(u,v) \le w_{uv}\). Furthermore, we have \[w_{u\bar v} = w_{uv} - w_{v\bar v} = w_{uv} - w_{u\bar u} = w_{\bar u v} \enspace .\]
	 Let \(\bar p\) be the farthest point from \(p = p(\lambda)\) on \(uv\). 
    \begin{enumerate}[C{a}se (1):]
		\item Let \(p\) be on \(u\bar{v}\) with \(p\ne\bar v\).
		
		We show that \(d(p,\bar p) = d(u, \bar u) = d(v, \bar v)\) and that \(\bar p\) moves from \(\bar u\) to \(v\) as \(p\) moves from \(u\) to \(\bar v\).
		
		This case requires \(d(u,v) < w_{uv}\), as otherwise \(u = \bar v\). Let \(\pi_{uv}\) be a shortest path from \(u\) to \(v\). The path \(\pi_{uv}\) and the edge \({uv}\) form a simple cycle of length \(w_{uv} + d(u,v)\). The farthest point \(\hat p\) from \(p\) on this cycle has network distance \(d(p,\hat p)=\frac{w_{uv} + d(u,v)}{2} = d(u,\bar u) = d(v,\bar v)\) to \(p\). The points \(u\), \(p\), \(\bar v\), and \(\bar u\) appear in this order along this cycle. Therefore, \(\hat p\) appears between \(\bar u\) and \(v\), which shows \(\hat p = \bar p\).
		
	%As the above argumentation applies to \(p=u\) as well, we obtain \(d(u,\bar u) = d(p,\bar p)\) and thus \(w_{up} = w_{\bar u\bar p}\). Together with \(w_{u\bar v} = w_{\bar uv}\), we can establish the claim about the exact position of \(\bar p\). Let \(\mu,\mu' \in [0,1]\) be such that \(p = (1-\mu) u + \mu \bar v\) and \(\bar p = (1-\mu')\bar u + \mu' v\). Then we have \(\mu = \frac{w_{up}}{w_{u\bar v}} = \frac{w_{\bar u\bar p}}{w_{\bar uv}} = \mu'\). 
		\item Let \(p\) be on \(\bar v c\).
		
		We show that \(v\) is farthest from \(p\) on \(uv\), i.e., \(\bar p = v\).
		
		If we walk from \(u\) to \(p\) along \(uv\), the distance to \(u\) increases from \(0\) to \(w_{up}\), until we reach \(\bar u\). Hence, for all points \(q \in up\) we have \(d(p,q) \le d(p,u) = w_{up}\). If we walk from \(v\) to \(p\) along \(uv\), the distance to \(v\) increases from \(0\) to \(w_{vp}\), until we reach \(\bar v\).  Hence, for all points \(q' \in pv\) we have \(d(p,q') \le d(p,v) = w_{pv}\). Since \(p \in uc\), we have \(w_{up} \le w_{pv}\) and, thus, infer that \(\bar p = v\) and \(d(p,\bar p ) = w_{pv}\). 
        
\end{enumerate}
The cases, \(p \in c\bar u\) and \(p \in \bar uv\), are symmetric to Case (1) and (2), respectively, because \(w_{u\bar v} = w_{\bar uv}\). In summary,
\begin{align*}
			\phi_{uv}^{uv}(\lambda)  = \max_{q \in uv} d(p(\lambda) , q) =
			\begin{dcases}
				d(u,\bar u) &\text{, if } p(\lambda) \in u\bar v\\
				w_{p(\lambda)v} &\text{, if } p(\lambda) \in \bar vc \\
				w_{up(\lambda)}  &\text{, if } p(\lambda) \in c\bar u \\
				d(v,\bar v) &\text{, if } p(\lambda) \in \bar uv\\
			\end{dcases},
\end{align*}
which implies the claim, since \(w_{p(\lambda)v} = (1-\lambda) w_{uv}\), \(w_{up(\lambda)} = \lambda w_{uv}\), and \(d(u,\bar u) = d(v,\bar v) = \frac{w_{uv} + d(u,v)}{2}\).
\end{proof}

%\begin{corollary} \label{thm::shape_of_phi_uv_st}
%	Given an edge \(uv\) in a network \(G=(V,E)\). All functions \(\phi_{uv}^{st}\), \(st \in E\), have a common slope of \(w_{uv}\) on their increasing segments and a common slope of \(-w_{uv}\) on their decreasing segments. 
%\end{corollary}

The eccentricity along an edge \(uv\) of a network with \(m\) edges is the upper envelope of \(m-1\) functions of the form described in \cref{thm::edge_to_edge_network_distance_function,fig::plot::edge_to_edge_distance} and one function of the form described in  \cref{thm::edge_to_self_network_distance_function,fig::plot::edge_to_self_distance}. Thus, it is continuous and piece-wise linear. Next, we bound the number of linear pieces. To break ties, we number the functions \(\phi_{uv}^{st}\) and say that the function with higher number is higher wherever two functions coincide. %Recall that we defined \(\phi_{uv}^{st}\), for two edges \(uv\) and \(st\), as the mapping from \(\lambda\in [0,1]\) to the distance from \(p = (1-\lambda) u + \lambda v\) to the farthest point from \(p\) on edge \(st\). 

\begin{lemma} \label{thm::segments_on_upper_envelope} 
Let \(uv\) be an edge of a network \(G\), and let \(k\) be the number of edges containing farthest points from some point on \(uv\). The eccentricity on \(uv\) consists of \(\Oh(k)\) segments.%\(4k-1\) line segments.
\end{lemma}
\begin{proof}%
%First, we argue why we can assume that the upper envelope of the functions \(\phi_{uv}^{st}\) begins and ends with (possibly empty) non-constant segments and has at most two non-constant segments. The only violation occurs when \(\phi_{uv}^{uv}\) has four non-empty segments and one of its constant segments appears at the beginning or end of the upper envelope. From the proof of \cref{thm::edge_to_self_network_distance_function}, we know that, in this case, there exists a simple cycle of length \(\frac{w_{uv} + d(u,v)}{2}\) containing the edge \(uv\) and we know that \(d(u,v) < w_{uv}\). Therefore, the eccentricity along \(uv\) is at least \(\frac{w_{uv} + d(u,v)}{2}\) and neither the descending nor the ascending segment of \(\phi_{uv}^{uv}\) appear on the upper envelope. 
%
For each non-constant segment only the part above the highest intersection with the other segments may appear on the upper envelope, due to the common slopes of the segments. Thus, each of the at most \(2k\) non-constant segments contributes at most two bending points to the upper envelope. Since there is no intersection between two constant segments, this accounts for all bending points of the eccentricity function on \(uv\), except for the first and the last one. Therefore, the eccentricity function on \(uv\) has at most \(4k+2\) bending points and, thus, consists of at most \(4k+1 = \Oh(k)\) segments.
\end{proof}

%We can employ, for instance, \citeauthor{hershberger1989finding}'s Algorithm~\cite{hershberger1989finding} to compute the upper envelope: %Nonetheless, we can establish a similar result in terms of upper and lower bounds on the asymptotic running time and space requirement using the plane-sweep paradigm. 
%\begin{lemma} \label{thm::eccentricity_function_on_edge}
%	Let \(uv\) be an edge of a network \(G\) with  \(m\) edges. Assume, we are given the functions \(\phi_{uv}^{st}\) for all edges \(st\) of \(G\). Then we can compute the upper envelope of these functions in \(O(m\log n)\) time.
%\end{lemma}

\subsection{The Eccentricity Diagram}

Due to the piece-wise linearity of the eccentricity, it suffices to state its value at the points corresponding to the endpoints of linear segments of the upper envelope of the functions \(\phi_{uv}^{st}\). This describes the eccentricity on the entire network: For a point \(p\) with \(p = (1-\lambda) a + \lambda b\) on a sub-edge \(ab\) with linear eccentricity we have 
\(
    \ecc(p) = (1-\lambda) \ecc(a) + \lambda \ecc(b) 
\). 
This leads us to the following notion, which is illustrated in \cref{fig::sunlet::example_eccentricity_diagram}. 
\begin{definition}[Eccentricity Diagram]\label{def::eccentricity_diagram} 
    Let \(G\) be a network. We call the subdivision of \(G\) with linear eccentricity on every edge and with the minimum number of vertices the \emph{eccentricity diagram} of \(G\) and denote it by \(\ED(G)\). 
\end{definition}
\begin{figure}[!ht]
	\centering
	\begin{subfigure}{0.32\linewidth}
		\centering
		\includegraphics[scale=0.66,page=1]{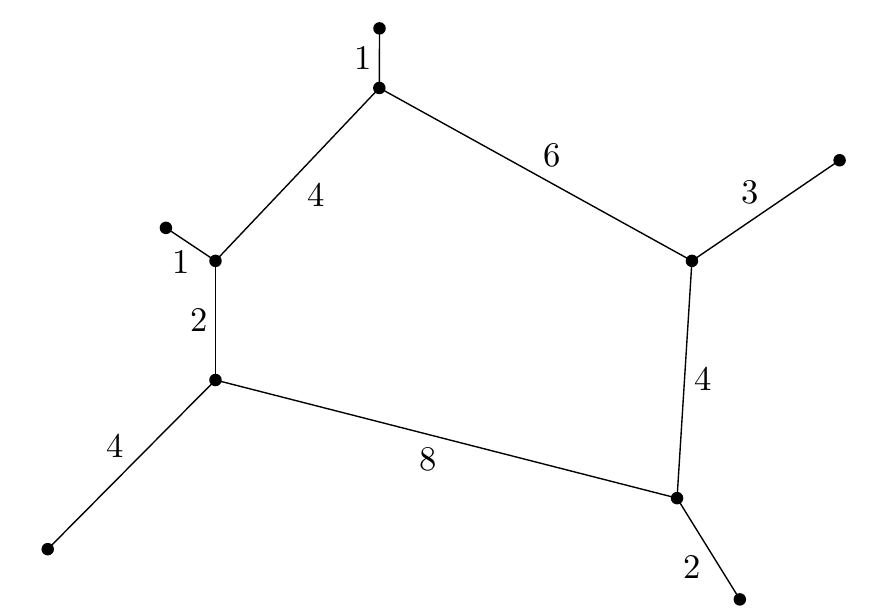}
		\caption{A network. \label{fig::sunlet::network}}
	\end{subfigure}%
	\begin{subfigure}{0.32\linewidth}
		\centering
		\includegraphics[scale=0.66,page=2]{sunlet}
		\caption{Its eccentricity diagram. \label{fig::sunlet::eccentricity_diagram}}	
	\end{subfigure}%
	\begin{subfigure}{0.32\linewidth}
		\centering
		\includegraphics[scale=0.66,page=3]{sunlet}
		\caption{Eccentricity values.\label{fig::sunlet::eccentricity_values}}
	\end{subfigure}
	\caption{From left \subref{fig::sunlet::network} to right \subref{fig::sunlet::eccentricity_values}: A (geometric) network \(G\),  its eccentricity diagram \(\ED(G)\) with sub-edge weights, and its eccentricity diagram with the eccentricity at each vertex of \(\ED(G)\). \label{fig::sunlet::example_eccentricity_diagram}}
\end{figure}
The eccentricity diagram of a network is well-defined and unique, as it can be obtained by subdividing each edge \(uv\) at the finitely many endpoints of the line segments of the eccentricity function on \(uv\). This yields a finite subdivision with the minimum number of additional vertices. An example is shown in \cref{fig::example_brute_force::eccentricity_diagram}.
\begin{figure}[!ht]
	\centering
	\begin{subfigure}{0.54\linewidth}
		\begin{subfigure}{\linewidth}
		\centering
		\includegraphics[page=1,scale=0.8]{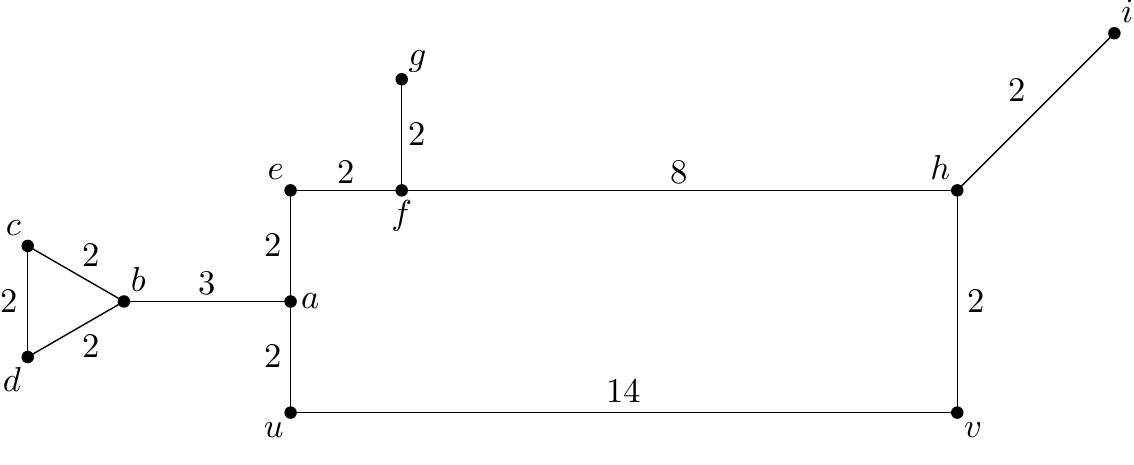}
		\caption{A network.} \label{fig::example_brute_force::input}
		\end{subfigure}
		
		\begin{subfigure}{\linewidth}
		\centering
		\includegraphics[page=4,scale=0.8]{example_brute_force}
		\caption{The subdivision of \(uv\).} \label{fig::example_brute_force::eccentricity_diagram::uv}
	\end{subfigure}%
	\end{subfigure}
	\begin{subfigure}{0.44\linewidth}
	\centering
		\includegraphics[]{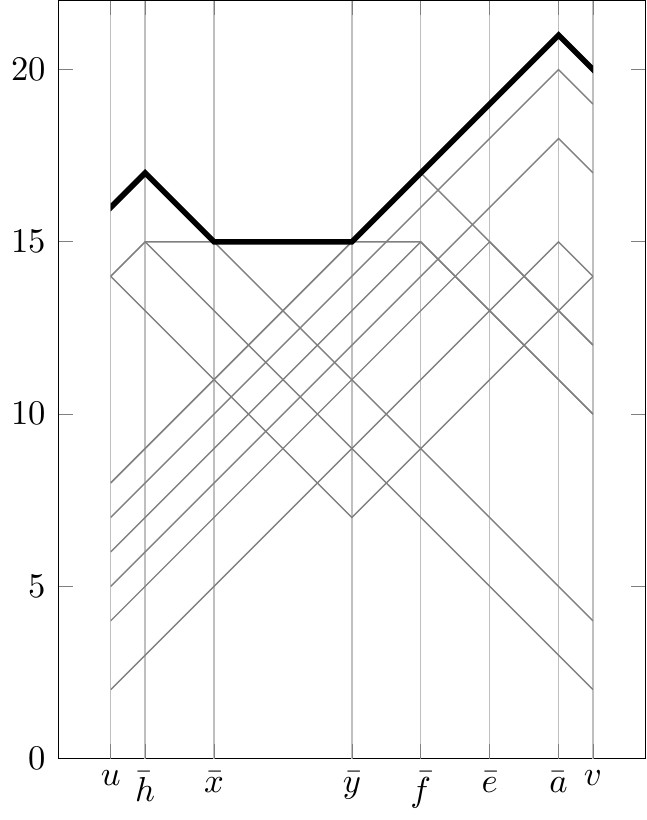}
		\caption{The functions \(p(\lambda) \mapsto \phi_{uv}^{st}(\lambda) = \max_{q \in st} d(p(\lambda),q)\), where \(p(\lambda) = (1-\lambda) u + \lambda v\) and \(\lambda \in [0,1]\), together with their upper envelope.} \label{fig::plot::example_brute_force::upper_envelope}
	\end{subfigure}

	\caption{An example of how the edge \(uv\) of the network in \subref{fig::example_brute_force::input} is subdivided \subref{fig::example_brute_force::eccentricity_diagram::uv} in the eccentricity diagram. In \subref{fig::plot::example_brute_force::upper_envelope} we see a plot of the functions \(p(\lambda) \mapsto \phi_{uv}^{st}(\lambda)\) from points \(p\) on \(uv\) to the network distance to their farthest points on each edge \(st\) of the network. The upper envelope (thick, black) of these functions is the eccentricity on \(uv\). We subdivide the edge \(uv\) at the points \(\bar h\), \(\bar x\), \(\bar y\), and \(\bar a\) to achieve linear eccentricity on each sub-edge.} \label{fig::example_brute_force::eccentricity_diagram}
\end{figure} 

%Computing the eccentricity along every edge as described, \cref{thm::segments_on_upper_envelope,thm::eccentricity_function_on_edge} yields the following \lcnamecref{thm::eccentricity_diagram_construction}.

\begin{lemma} \label{thm::eccentricity_diagram_construction}
    The eccentricity diagram of a network \(G = (V,E)\) with \(n\) vertices and \(m\) edges has size \(\mathcal{O}(m^2)\) and can be constructed in \(\mathcal{O}(m^2 \log n)\) time.
\end{lemma}
\begin{proof}
    The eccentricity function along each edge \(uv\) consists of \(\Oh(k) = \Oh(m)\) segments, where \(k\) is the number of edges containing farthest points from \(uv\). As we subdivide each of the \(m\) edges at the bending points of the eccentricity function, we obtain the bound on the size. 

    We can construct the eccentricity diagram as follows. First, we compute the network distances between all pairs of vertices of \(G\). We can use, for instance,  \citeauthor{johnson1977efficient}'s all-pairs shortest path algorithm~\cite{johnson1977efficient}, which has a running time of \(\Oh(n^2\log(n) +nm) = \Oh(m^2 \log n)\). With this information and \cref{thm::point_to_edge_network_distance,thm::edge_to_edge_network_distance_function,thm::edge_to_self_network_distance_function}, we can determine the functions \(\phi_{uv}^{st}\) for all edges \(uv,st\in E\). Then we can compute the upper envelope of the functions \(\phi_{uv}^{st}\) over all edges \(st \in E\). For instance, \citeauthor{hershberger1989finding}'s Algorithm~\cite{hershberger1989finding} can accomplish this task in \(\Oh(m \log n)\) time per edge \(uv\). The overall construction time is \(\Oh(m^2 \log n)\). 
\end{proof}

%If each edge is processed independently, this algorithm may take \(\Omega(m^2\log n)\) time. %To obtain a bad instance, we can create \(k\) copies of the vertex \(v\) in \cref{fig::construction_upper_envelope_lower_bound}; the resulting network will have \(7k\) edges. The computation of the upper envelope will then take \(\Omega(k\log(k))\) time on each of the \(k\) copies of the edge \(uv\). 
%However, this implies no lower bound for the problem itself.%
%
Next, we establish a lower bound on the size of eccentricity diagrams. The corresponding construction in the proof of \cref{thm::eccentricity_diagram_lower_size_bound} below shows that the bound stated in \cref{thm::eccentricity_diagram_construction} is tight for planar networks. 

\begin{lemma} \label{thm::eccentricity_diagram_lower_size_bound}
 For every \(k,l\in \N\) with \(k \ge 2\) and \(0 < l \le k^2\), there is a network with \(4k\) vertices and \(4k-1+l\) edges whose eccentricity diagram has size \(\Omega(kl)\).  
\end{lemma}

\begin{proof} Let \(k \ge 2\) and let \(\epsilon  \) be such that \(0 < \epsilon < \frac{1}{k-1}\).  Consider the network \(G_{k,0}\) formed by the black edges of the network depicted in \cref{fig::construction_eccentricity_diagram_size_base}. We obtain the network \(G_{k,l}\) by adding \(l\) edges of the form \(v_iu_j\) with \(i,j \in \set{1,2,\dots,k}\) and weight \(w_{v_iu_j} = k-1\) to \(G_{k,0}\).  All edge weights are positive, since \(k\ge 2\). Thus, the network distance is a metric on \(G_{k,l}\). 
\begin{figure}[!ht]
	\centering
	\includegraphics[scale=1]{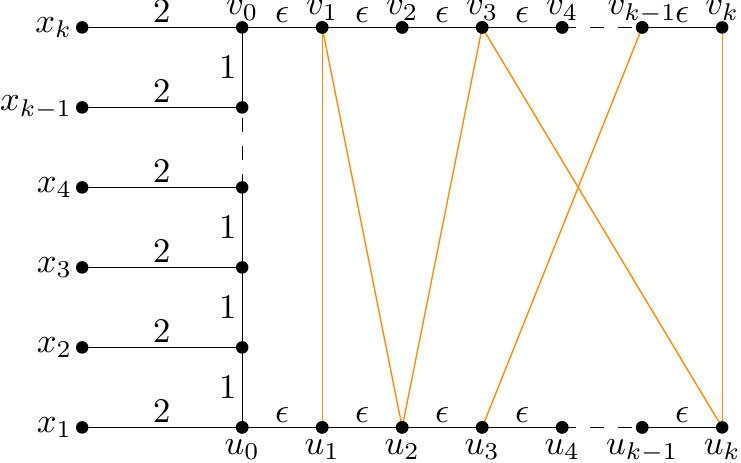}
	\caption{The black edges form the network \(G_{k,0}\). Adding \(l\) orange edges \(v_iu_j\) with edge weight \(w_{v_iu_j} = k-1\) and  \(i,j\in \set{1,2,\dots,k}\) yields a network \(G_{k,l}\).} \label{fig::construction_eccentricity_diagram_size_base}
\end{figure}
\begin{figure}[!ht]
    \centering
    \begin{subfigure}{0.48\linewidth}
		\centering 
		\includegraphics[scale=1]{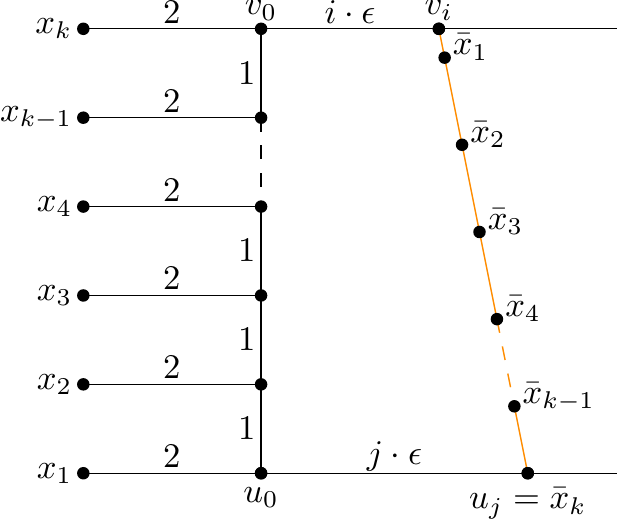}
		\caption{The farthest point \(\bar x_r\) from vertex \(x_r\) on the edge \(v_iu_j\) with \(i\le j\) for all \(r=1,2,\dots,k\).}
		 \label{fig::construction_eccentricity_diagram_size_comp_x}
    \end{subfigure}\quad%
    \begin{subfigure}{0.48\linewidth}
	\centering
		\includegraphics{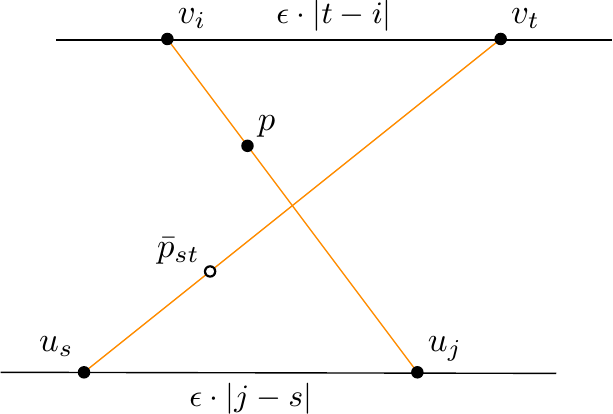}
		\caption{The farthest point \(\bar p_{st}\) from \(p \in v_iu_j\) on another non-\(G_{k,0}\) edge \(u_sv_t\).\label{fig::construction_eccentricity_diagram_size_comp_l}} 
	\end{subfigure}
    \begin{subfigure}{1\linewidth}
    	\centering
	\includegraphics{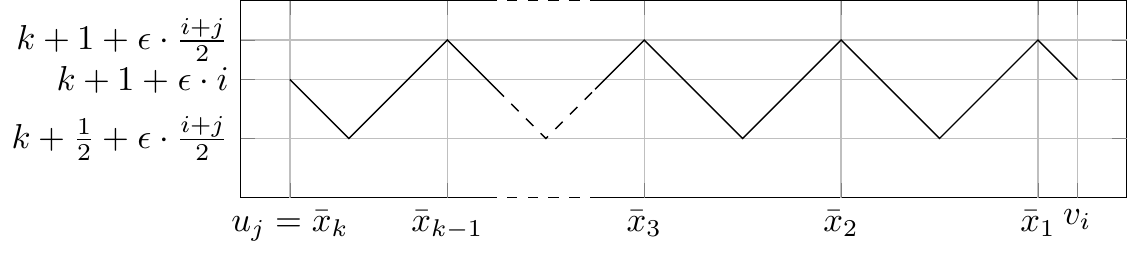}
	\caption{The distance from any point \(p\) on \(u_jv_i\) to the farthest among \(x_1,x_2,\dots,x_k\).}
	\label{fig::construction_eccentricity_diagram_size_envelope} 
	\end{subfigure}
	\caption{An illustration of the arguments in the proof of \cref{thm::eccentricity_diagram_lower_size_bound}. Any non-\(G_{k,0}\) edge \(u_jv_i\) in network \(G_{k,l}\) is subdivided \subref{fig::construction_eccentricity_diagram_size_comp_x} into \(k\) edges in the eccentricity diagram of \(G_{k,l}\). We show that the distance from any point \(p\) on \(u_jv_i\) to the farthest point \(\bar p_{st}\) on another non-\(G_{k,0}\) edge \(u_sv_t\) \subref{fig::construction_eccentricity_diagram_size_comp_l} is smaller than the distance to the farthest point among \(x_1, x_2, \dots, x_k\) \subref{fig::construction_eccentricity_diagram_size_envelope}.  \label{fig::construction_eccentricity_diagram_size_arguments}}
\end{figure}

 \Cref{fig::construction_eccentricity_diagram_size_arguments} illustrates the following arguments. First, consider a non-\(G_{k,0}\) edge \(v_iu_j\) with \(i\le j\). We show that \(v_iu_j\) will be subdivided into at least \(2k-2\) sub-edges in the eccentricity diagram of \(G_{k,l}\), as shown in \cref{fig::construction_eccentricity_diagram_size_comp_x}. For each \(r=1,2,\dots,k\), let \(\bar{x}_r\) be the point on \(v_iu_j\) that is farthest from \(x_r\). Then we have
\(
	d(x_r,\bar x_r)= k+1 + \epsilon\frac{i+j}{2} 
\), 
for each \(r=1,2,\dots,k-1\) and \(\bar x_k = u_j\) with
\(
	d(x_k,\bar x_k) = d(x_k,u_j) = k+1+\epsilon i
\).

A plot of the mapping \(p \mapsto \max_{r=1}^kd(p,x_r)\) from the points \(p\) on \(u_jv_i\) to their farthest point(s) among \(x_1\), \(x_2\), \dots, \(x_r\) is shown in \cref{fig::construction_eccentricity_diagram_size_envelope}. We claim that this function coincides with the eccentricity function, i.e., that \(x_1\), \(x_2\), \dots, \(x_r\) are the only eccentric points from the edge \(u_jv_i\): We observe that the maximum distance to any of the vertices \(x_1\), \(x_2\), \dots, \(x_r\)  is always at least \(k + \frac12 + \epsilon\cdot \frac{i+j}2 > k+\epsilon\). 

On the other hand, all points on the network have distance at most \(k\) from \(p\), except for points on edges incident to any of \(x_1,x_2,\dots,x_k\). Let \(\bar p_{st}\) be the farthest point from \(p\) on the edge \(u_sv_t\) as shown in \cref{fig::construction_eccentricity_diagram_size_comp_l}. Here we have \(d(p,\bar p_{st}) < k\) because 
\begin{align*}
	d(p,\bar p_{st}) &= \frac{d(p,u_s) + d(p,v_t) + w_{u_sv_t}}{2}, \\
	&= \frac{ w_{pu_j} + \epsilon\cdot \abs{j-s} + w_{pv_i} + \epsilon\cdot \abs{i-t} + w_{u_sv_t}}{2}, \\
	&= \frac{w_{u_sv_t} + w_{u_jv_i} + \epsilon\cdot \left(\abs{j-s}+ \abs{i-t}\right) }{2}, \\
	&= k-1 + \epsilon \cdot \frac{\abs{j-s}+ \abs{i-t}}{2}, \\
	&\le k-1 + \epsilon \cdot(k-1) < k.
\end{align*}
Therefore, \cref{fig::construction_eccentricity_diagram_size_envelope} shows the eccentricity along \(v_iu_j\). If \(i=j\) then this function consists of \(2k-2\) segments and \(2k-1\) otherwise. This is true for any of the \(l\) non-\(G_{k,0}\) edges. Hence, there are at least \(4k-1 + l\cdot (2k-2) \in \Omega(lk)\) edges in the eccentricity diagram of \(G_{k,l}\).  
\end{proof}

\begin{corollary}\label{thm::eccentricity_diagram_many_edges_size_bound}  For all \(n, m \in \N\) with  \(m \in \omega(n)\) and \(n \ge 8\), there exists a network  with \(n\) vertices and \(m\) edges whose eccentricity diagram has size \(\Omega(nm)\).
\end{corollary} 
\begin{proof}
	Let \(n \ge 8\) and \(m\in \omega(n)\). Consider \(G_{k,0}\) from the proof of \cref{thm::eccentricity_diagram_lower_size_bound} with \(k = \left\lfloor{\frac{n}{4}}\right\rfloor\). By adding \(l= m-n+1\) edges of the form \(v_iu_j\), \(i,j=1,2,\dots,k\), we obtain a network \(G_{k,l}\) with \(k=\Theta(n)\) and \(l=\Theta(m)\). According to \cref{thm::eccentricity_diagram_lower_size_bound}, the eccentricity diagram of \(G_{k,l}\) has size  \(\Omega(k l) = \Omega(nm)\). Subdividing \(m \bmod 4\) edges of \(G_{k,l}\) yields a network with the same property and exactly \(n\) vertices and \(m\) edges.% exactly \(n\) vertices and \(m\) edges. 
\end{proof}

\begin{corollary} \label{thm::eccentricity_diagram_planar_size_bound} For all \(n \in \N\) with \(n\ge 8\), there exists a planar network  with \(n\) vertices whose eccentricity diagram has size \(\Omega(n^2)\).
\end{corollary} 
\begin{proof}
	Let \(n \ge 8\). Consider \(G_{k,0}\) from the proof of \cref{thm::eccentricity_diagram_lower_size_bound} with \(k = \left\lfloor{\frac{n}{4}}\right\rfloor \in \Theta(n)\). By adding all edges of the form \(v_iu_i\) for \(i=1,2,\dots,k\), we obtain a planar network \(G_{k,k}\). According to \cref{thm::eccentricity_diagram_lower_size_bound}, the eccentricity diagram of this network has size  \(\Omega(k^2) = \Omega(n^2)\). Subdividing \(n \bmod 4\) edges of \(G_{k,k}\) yields a network with the same property and exactly \(n\) vertices.
\end{proof}

\subsection{A Data Structure for Eccentricity Queries}

Assume we are given the eccentricity diagram \(\ED(G)\) of a network \(G\) as well as the eccentricity \(\ecc(x)\) of each vertex \(x\) of \(\ED(G)\). Then we can answer queries for the eccentricity \(\ecc(p)\) of any point \(p\) on \(G\): It suffices to identify the edge \(ab\) of \(\ED(G)\) containing \(p\), since 
\[
    \ecc(p) = \left(1 - \frac{w_{ap}}{w_{ab}} \right) \ecc(a) + \frac{w_{ap}}{w_{ab}} \ecc(b).
\]
Recall that eccentricity queries consist of a point \(p\) on \(G\) and the edge \(uv\) containing \(p\). For each edge \(uv\) of \(G\), we store the vertices \(x \in uv\) of \(\ED(G)\) (e.g., in an array or in a balanced binary search tree) sorted by the fraction \(\frac{w_{ux}}{w_{uv}}\) at which \(x\) subdivides \(uv\). Then we can find the sub-edge \(ab\) containing \(p\) with a binary search for \(\frac{w_{up}}{w_{uv}}\) in \(\Oh(\log n)\) time, as there are at most \(\Oh(m) =\Oh(n^2)\) vertices of \(\ED(G)\) on \(uv\). 
\begin{theorem} \label{thm::eccentricity_diagram_data_structure}
    Given a network \(G\) with \(n\) vertices and \(m\) edges. There is a data structure that can be constructed in \(\Oh(m^2 \log n)\) time and has size \(\Oh(m^2)\) supporting queries for the eccentricity of any point \(p\) on \(G\) in \(\Oh(\log n)\) time, provided that the edge \(uv\) of \(G\) containing \(p\) is given. 
\end{theorem}

%\Cref{thm::eccentricity_diagram_data_structure} requires the distances between all vertices of the network; we use them to determine the functions \(\phi_{uv}^{st}\) throughout the construction of the eccentricity function. In practice, the computation of these distances may be the bottleneck when constructing our data structure for eccentricity queries. In theory, however, we can rely on \citeauthor{johnson1977efficient}'s all-pairs shortest path algorithm~\cite{johnson1977efficient}, which takes \(\Oh(n^2\log(n) +nm) \subseteq \Oh(m^2 \log n)\) time. 

%Depending on the type of network, this may or may not be the bottleneck in the construction time of the data structure in \cref{thm::eccentricity_diagram_data_structure}. For instance, if the network is planar we can use the \(\Oh(n^2)\) time algorithm by \textcite{frederickson1991planar} or if the network is sparse the \(\Oh(n^2\log(n) +nm)\) time algorithm by \textcite{johnson1977efficient}. Otherwise we may resort to the \(\Oh(n^3)\) time algorithm by \textcite{floyd1962algorithm}.  \textcite{chan2010more} provides an improved result for dense networks in \(\Oh(n^3 \log^3 \log n/ \log^2 n)\) time and a summary of recent results regarding all pairs shortest path algorithms. 
	
\section{Network Farthest-Point Diagrams} \label{sec::network_farthest_point_diagrams}

Apart from the network distance towards farthest points, we are also interested in their location: We seek to query for the set of all farthest points from any point \(g\) on a network \(G\). This suggests that we subdivide the network into parts with a common set of farthest points and then find the part containing \(g\). In order to analyze this approach, we formally define this subdivision, which is the farthest-point Voronoi diagram whose metric space consists of all points on \(G\) and the network distance \(d_G(\cdot, \cdot)\) and whose sites are all points on \(G\).

\subsection{Farthest-Point Network Voronoi Link Diagrams}

\begin{definition}[Farthest-Point Network Voronoi Link Diagram]
Let \(G\) be a network. 
\begin{enumerate}[(i)]
\item Let \(g\) be a point on \(G\). The set of points on \(G\) to whom \(g\) is a farthest point is denoted by \(\mathcal{V}_{\text{far-net}}(g)\), i.e.,
\[
        \mathcal{V}_{\text{far-net}}(g) := \set{p \in G \colon \forall g'  \in G \colon d(p,g') \le d(p,g)}.
\]
We call $\mathcal{V}_{\text{far-net}}(g)$ the \emph{farthest-point network Voronoi link cell} of \(g\). 
\item We obtain the \emph{farthest-point network Voronoi link diagram} of \(G\) by subdividing \(G\) at each boundary point of the non-empty farthest-point network Voronoi link cells, i.e., at all points in the set 
    \(
       \mathcal{S} \coloneqq \bigcup_{g \in G} \partial \mathcal{V}_{\text{far-net}}(g)
    \).
    \item We say that a farthest-point network Voronoi link diagram is \emph{finite} if and only if \(\mathcal{S}\) is finite.
\end{enumerate}
\end{definition}

Finite farthest-point network Voronoi link diagrams are, by definition, (finite) networks with finitely many vertices. Infinite farthest-point network Voronoi link diagrams, on the other hand, are \emph{infinite networks} with infinitely many vertices and degenerate edges that are reduced to single points, as shown in \cref{fig::example_twocycle_cfpnvd}. Traditionally, Voronoi diagrams are determined by a finite set of sites. The \emph{farthest-point Voronoi diagram} \cite[Section~3.3]{oktabe2000spatial} subdivides the plane into regions with a common farthest point among finitely many points in the plane. The \emph{network Voronoi link diagram} \cite[Section~3.8]{oktabe2000spatial} subdivides a network into parts with common closest vertices among finitely many vertices. The farthest-point network Voronoi link diagram differs from other Voronoi diagrams in at least two ways. First, we are unaware which points are the sites, i.e., which points on a network are farthest points. Second, there are infinitely many farthest points when the farthest-point network Voronoi link diagram is infinite, as depicted in \cref{fig::examples_cfpnvd}. We cannot immediately apply known methods to produce farthest-point network Voronoi link diagrams, and we cannot use these diagrams for farthest-point-set queries in the infinite case.%
\begin{figure}[!hb]
\centering
\begin{subfigure}[b]{.49\linewidth}
\centering
\includegraphics{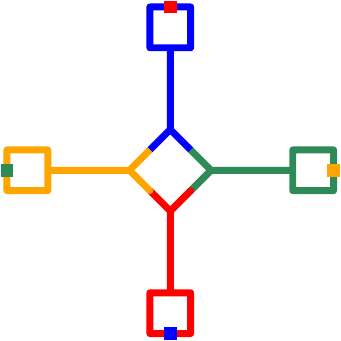}
\caption{Finite diagram.} \label{fig::example_quadlet_ecc_diag}
\end{subfigure} 
\begin{subfigure}[b]{.49\linewidth} \centering
\includegraphics{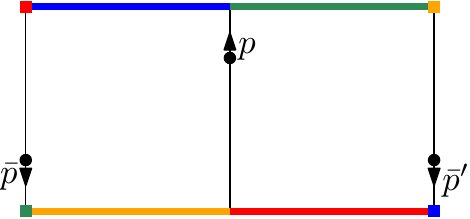}
\caption{Infinite diagram.} \label{fig::example_twocycle_cfpnvd}
\end{subfigure}
\caption{The farthest-point network Voronoi link diagrams for two networks. Parts that have a common farthest point (square) are indicated in colour. In the finite case \subref{fig::example_quadlet_ecc_diag}, the network is subdivided into regions with a fixed farthest point. In the infinite case  \subref{fig::example_twocycle_cfpnvd}, we have a different behaviour on the vertical edges (black): When the point \(p\) is moved upwards, its two farthest points \(\bar p\) and \(\bar p'\) move downwards accordingly. No two points on this edge have a common farthest point.\vspace*{-0.5cm}} \label{fig::examples_cfpnvd}
\end{figure}%

Characterizing the finiteness of farthest-point network Voronoi diagrams reveals how we can avoid the difficulties of the infinite case. The following auxiliary lemma will help us with this characterization. 

\begin{lemma}  \label{thm::no_multiple_cells_along_edges}
    Given an edge \(uv\) of a network \(G\) and a point \(g\) on \(G\). The set of points on \(uv\) that have \(g\) as a farthest point, i.e., the set \(\mathcal{V}_\text{far-net}(g) \cap uv\), is a (possibly empty) interval on \(uv\). 
\end{lemma}
\begin{proof}
	Assume, for the sake of a contradiction, that the statement is false. Then, there are points \(a\) and \(b\) on \(uv\) such that \(g\) is eccentric to \(a\) and \(b\) but not to some point \(c\) on \(ab\). Let \(\bar g'\) be a farthest point from \(c\) in \(G\). \Cref{fig::sketch_cell_interval_on_edge} shows a sketch of this (impossible) constellation. 
	\begin{figure}[ht]
		\centering
		\includegraphics{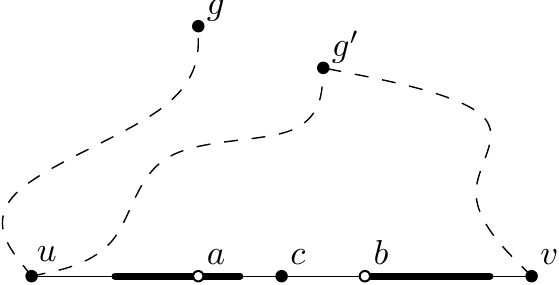}
		\caption{The impossible constellation where the points on an edge \(uv\) that have \(g\) as farthest point consists of two disjoint intervals (thick, black). There would be two points \(a\) and \(b\) in these intervals such that the farthest point \(\bar c\) of the midpoint \(c\) of \(ab\) is different from \(g\).} \label{fig::sketch_cell_interval_on_edge}
	\end{figure}
	
	We first argue why \(g\) cannot be contained in \(ab\). \Cref{thm::edge_to_self_network_distance_function} implies that if \(g \in uv\) and if \(g\) is farthest from both \(a\) and \(b\), then we have either \(g=u\) or \(g=v\). Recall that no two points on sub-edges with constant value of \(\phi_{uv}^{uv}\)  have the same farthest points on \(uv\), and that all points on the sub-edges of \(uv\) with increasing (respectively decreasing) value of \(\phi_{uv}^{uv}\) have \(v\) (respectively \(u\)) as their farthest point on \(uv\). We have \(g\notin ab\) in either case. 
	
%	More precisely, \(a\) and \(b\) are either both contained in the ascending segment or both contained in the descending segment of the function \(\phi_{uv}^{uv}\). 
	
	Let \(\bar g\) be the farthest point from \(g\) on \(ab\). Without loss of generality, \(c\) is located on \(a\bar g\). As we walk from \(a\) to \(c\) along \(uv\), the network distance to any point \(q\) in the network can change by at most \(w_{ac}\), i.e., for all \(q\) on \(G\) we have 
	\(
		\abs{d(c,q) - d(a,q)} \le w_{ac}
	\). The distance to \(g\) increases by exactly this amount as the network distance to \(g\) is increasing on the sub-edge \(a\bar g\), which includes the sub-edge \(ac\), i.e, \(d(c,g) - d(a,g) = w_{ac}\).
	
	We assume that \(g\) is eccentric to \(a\) and \(b\) but not to \(c\), i.e., \(d(c,g) < d(c,g')\), \(d(a,g') \le d(a,g) = \ecc(a)\), and \(d(b,g') \le d(b,g) = \ecc(g)\) for some point \(g'\) on \(G\). Therefore, the increase in the network distance to \(g'\) on \(ac\) must be strictly higher than the increase in network distance to \(g\) on \(ac\), i.e.,
	\(
		d(c,g') - d(a,g') > d(c,g) - d(a,g)
	\).
	 This contradicts the assessment that the network distance to \(g\) already achieves the maximum possible increase of \(w_{ac}\) along \(ac\). Therefore, \(g'\) cannot exist and the claim follows.
\end{proof}

\begin{theorem} \label{thm::characterization_of_infiniteness} 
Let \(G\) be a network. The farthest-point network Voronoi link diagram of \(G\) is finite if and only if there is no edge \(ab\) in the eccentricity diagram \(\ED(G)\) of \(G\) such that the eccentricity is constant on \(ab\).%, i.e., we have \(\ecc(a) =  \ecc(p)\) for all \(p\in ab\). 
\end{theorem}
\begin{proof} Let the farthest-point network Voronoi link diagram of \(G\) be finite.  First, we show that there are only finitely many points on \(G\) that are farthest from some point on \(G\). We then infer that no sub-edge with constant eccentricity exists in \(G\), and thus, no edge with constant eccentricity in \(\ED(G)\).

Each of the finitely many boundary points of the farthest-point network Voronoi link cell has at most one farthest point per each edge of \(G\). Since there is no change to the set of farthest points within the farthest-point network Voronoi link cells, there are only finitely many points on \(G\) that are farthest from some point on \(G\). Let \(g_1, g_2, \dots, g_k\) be these points. As we walk along an edge \(uv\) of \(G\), the network distance to \(g_i\) strictly increases from \(u\) to \(\bar g_i\), the farthest point from \(g_i\) on \(uv\), and then strictly decreases until \(v\), as illustrated in \cref{fig::sketch_characterization_infinity_ii_to_iii}. Since the eccentricity on \(uv\) is the upper envelope of the network distances to the points \(g_1, g_2, \dots, g_k\), the eccentricity cannot be locally constant on \(uv\) and there is no sub-edge of \(G\) with constant eccentricity.% Consequently, no edge of \(\ED(G)\) has constant eccentricity.
\begin{figure}[ht]
\centering
\begin{subfigure}[b]{0.49\linewidth}
	\centering
	\includegraphics[scale=0.8]{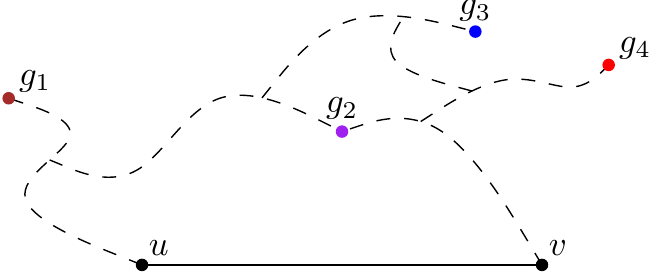}
	\caption{Network}\label{fig::sketch_characterization_infinity_ii_to_iii_network}
\end{subfigure}
\begin{subfigure}[b]{0.49\linewidth}\centering
	\includegraphics[scale=0.8]{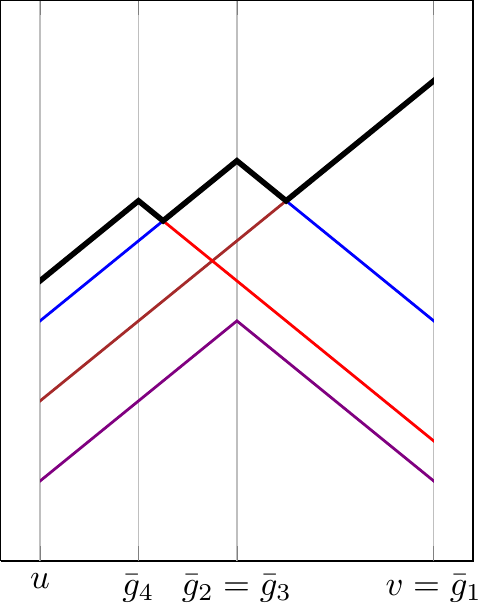}
	\caption{Distances and envelope.}\label{fig::sketch_characterization_infinity_ii_to_iii_distances}
\end{subfigure}
\caption{On the left \subref{fig::sketch_characterization_infinity_ii_to_iii_network}, a sketch of a network  with finitely many farthest points \(g_1, g_2, g_3\), and \(g_4\). On the right \subref{fig::sketch_characterization_infinity_ii_to_iii_distances}, a plot  of the network distances (coloured) to the points \(g_i\) on edge \(uv\). The upper envelope (black) of these finitely many non-constant functions cannot be locally constant on \(uv\).}\label{fig::sketch_characterization_infinity_ii_to_iii}
\end{figure}
\item Conversely, let there be no edge in the eccentricity diagram of \(G\) with constant eccentricity. Let \(ab\) be an edge of \(\ED(G)\) that is a sub-edge of edge \(uv\) of \(G\). The eccentricity strictly increases or strictly decreases from \(a\) to~\(b\). Assume, without loss of generality, that the former is true as shown in \cref{fig::must_leave_uv_via_u}. 
\begin{figure}
	\centering
	\begin{subfigure}{0.45\linewidth}
	\centering
	\includegraphics[scale=0.8]{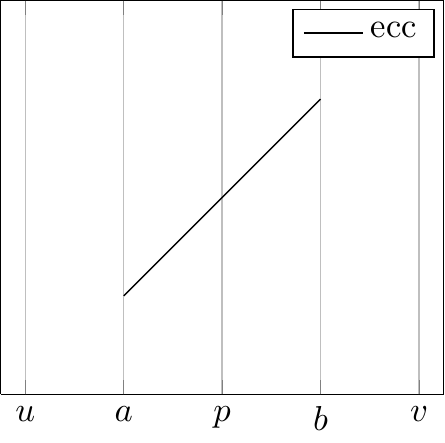}
	\end{subfigure}	
	\begin{subfigure}{0.45\linewidth}
		\centering 
		\includegraphics[scale=0.8]{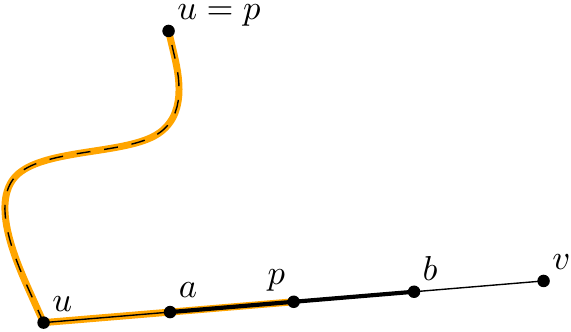}
	\end{subfigure}	
	\caption{A sub-edge \(ab\) of an edge \(uv\) of \(G\) with strictly increasing eccentricity. Any shortest path (orange) from a point \(p \in ab\) to a farthest point \(\bar p\) from \(p\) leaves \(uv\) via \(u\). Thus, \(\bar p\) is also a farthest point of \(u\).} \label{fig::must_leave_uv_via_u}
\end{figure}
Now, let \(p\) be a point on \(ab\), and let \(\bar p\) be a farthest point from \(p\) in \(G\). Every shortest path from \(p\) to \(\bar p\) leaves \(uv\) through \(u\). Thus, any farthest point from \(p\) is also a farthest point from \(u\). Therefore, there cannot be more non-empty farthest-point network Voronoi link cells than vertices of \(G\). Due to \cref{thm::no_multiple_cells_along_edges}, these have at most two boundary points in the interior of each edge. Hence, the number of boundary points is finite and so is the farthest-point network Voronoi link diagram.\end{proof}

\subsection{Network Farthest-Point Diagrams}
\Cref{thm::characterization_of_infiniteness} shows that there are two ways in which farthest points change as we move along an edge of a network. First, the farthest points can remain the same. This happens if we move along a sub-edge with ascending or descending eccentricity. Second, the farthest points can all move staying at the same distance. This happens if we move along a sub-edge with constant eccentricity. In both cases, the edges containing farthest points change at most finitely often. Knowing the edges containing farthest points suffices to reconstruct the farthest points. We obtain a finite representation of the farthest-point network Voronoi link diagram by subdividing the edges of the eccentricity diagram depending on which edges of the network contain farthest points. %An example of the resulting \emph{network farthest-point diagram} is shown in \Cref{fig::sunlet::example_network_farthest_point_diagram}.
\begin{definition}[Network Farthest-Point Diagram] \label{def::farthest_point_diagram} 
 Let \(G\) be a network. Consider the subdivisions of the eccentricity diagram \(\ED(G)\) of \(G\) such that the points in the  interior of every edge \(uv\) of the subdivision have a common set of edges of \(G\) containing their farthest points, i.e., for all edges \(e\) of \(G\) and all \(p,q \in uv \setminus \set{u,v}\)
\[ 
	\exists \bar p \in e \colon d(p,\bar p) = \ecc(p) \iff \exists \bar q \in e \colon d(q,\bar q) = \ecc(q).
\] 
Among these subdivisions, we call the one with the least number of additional vertices the \emph{network farthest-point diagram} of \(G\) and denote it by \(\FD(G)\). 
\end{definition}
%\begin{figure}[!ht]
%	\centering
%	\begin{subfigure}{0.499\linewidth}
%		\centering
%		\includegraphics[scale=0.9,page=1]{sunlet}
%		\caption{A network. \label{fig::sunlet::network::again}}
%	\end{subfigure}%
%	\begin{subfigure}{0.499\linewidth}
%		\centering
%		\includegraphics[scale=0.9,page=2]{sunlet}
%		\caption{The network farthest-point diagram. \label{fig::sunlet::network_farthest_point_diagram}}	
%	\end{subfigure}
%	
%	\begin{subfigure}{0.499\linewidth}
%		\centering
%		\includegraphics[scale=0.9,page=4]{sunlet}
%		\caption{The farthest points.\label{fig::sunlet::eccentric_points}}
%	\end{subfigure}%
%	\begin{subfigure}{0.499\linewidth}
%		\centering
%		\includegraphics[scale=0.9,page=5]{sunlet}
%		\caption{Farthest-point regions.\label{fig::sunlet::farthest_point_information}}
%	\end{subfigure}
%	\caption{An example for the network farthest-point diagram \(\FD(G)\) \subref{fig::sunlet::network_farthest_point_diagram} of a (geometric) network \(G\) \subref{fig::sunlet::network::again} together with a visualization of the farthest points \subref{fig::sunlet::farthest_point_information} and their location \subref{fig::sunlet::eccentric_points}. Line segments of same colour in \subref{fig::sunlet::eccentric_points} indicate that every point in the corresponding sub-edge in \subref{fig::sunlet::farthest_point_information} has a farthest point within this line segment. The eccentricity of the vertices of \(\FD(G)\) are shown in \subref{fig::sunlet::farthest_point_information}.\label{fig::sunlet::example_network_farthest_point_diagram}}
%\end{figure}
Network farthest-point diagrams are well-defined and unique: We subdivide an edge of the eccentricity diagram wherever the set of edges containing farthest points changes. This occurs at most \(\Oh(m)\) times, i.e., when one of the functions \(\phi_{uv}^{st}\) joins with or departs from the eccentricity function, as shown in \Cref{fig::example_brute_force_fps}.
\begin{figure}[ht]
	\centering
	\begin{subfigure}{0.59\linewidth}
		\begin{subfigure}{\linewidth}
		\centering
		\includegraphics[page=1,scale=0.8]{example_brute_force}
		\caption{A network.} \label{fig::example_brute_force::input::again}
		\end{subfigure}
		
		\begin{subfigure}{\linewidth}
		\centering
		\includegraphics[page=6,scale=0.8]{example_brute_force}
		\caption{The subdivision of \(uv\).} \label{fig::example_brute_force::network_farthest_point_diagram::uv}
	\end{subfigure}%
	\end{subfigure}
	\begin{subfigure}{0.39\linewidth}
	\centering
		\includegraphics[]{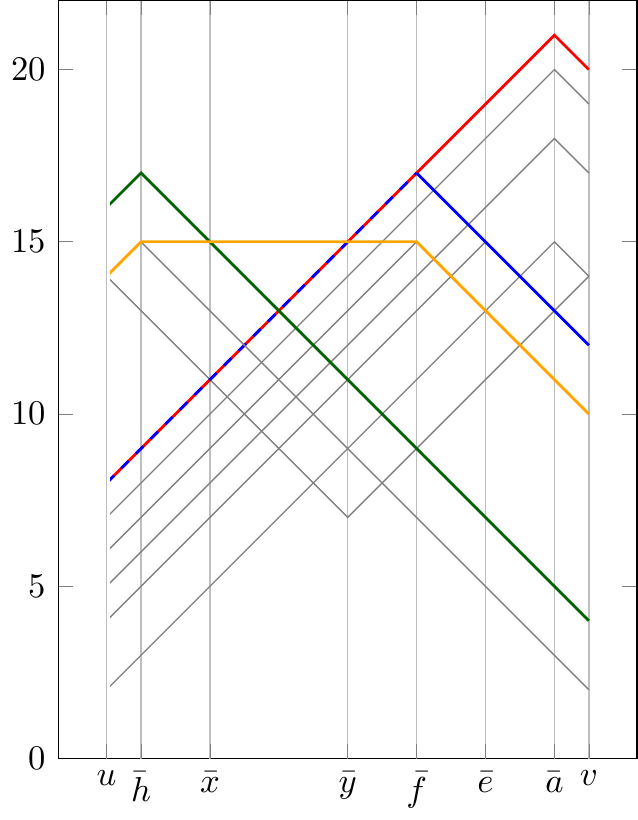}
		\caption{The functions \(\phi_{uv}^{st}\)  and their upper envelope.} \label{fig::plot::example_brute_force_fps::upper_envelope}
	\end{subfigure}
	
	\caption{An example for the network farthest-point diagram for the network in \subref{fig::example_brute_force::input::again}.  The upper envelope  of the functions \(\phi_{uv}^{st}\), which is shown in \subref{fig::plot::example_brute_force_fps::upper_envelope}, reveals which edges contain farthest points. On the left \subref{fig::example_brute_force::network_farthest_point_diagram::uv} shows the network farthest-point diagram on \(uv\). The farthest points from \(uv\) are located at the dot(ted segments) of matching colour: For instance, the sub-edge \(xy\) contains the farthest point for each point on the sub-edge \(\bar x \bar y\), and vertex \(i\) is the farthest point for all points on the sub-edge \(ux\). \label{fig::example_brute_force_fps}} 
\end{figure}

%\subsection{A Data Structure for Farthest-Point-Set Queries}

Storing the edges containing farthest points with every edge of the network farthest-point diagram yields a data structure for farthest-point-set queries of size \(\Oh(m^3)\), since there are \(\Oh(m)\) farthest points for each of the \(\Oh(m^2)\) edges of \(\FD(G)\). We answer a farthest-point-set query for a point \(p\) on  edge \(uv\) as follows. First, we identify the edge \(ab\) in \(\FD(G)\) containing \(p\), using binary search. A list of the edges containing farthest points from \(p\) is stored with \(ab\). For each edge \(st\) in this list, we compute the farthest point from \(p\) on \(st\). This takes constant time per edge using \cref{thm::point_to_edge_network_distance}. Thus, if \(p\) has \(k\) farthest points, we can report them in \(\Oh(k + \log n)\) time.  
%
%
%\begin{theorem} \label{thm::naive_brute_force_farthet_points_data_structure}
%    Given a network \(G\) with \(n\) vertices and \(m\) edges. There is a data structure that supports farthest-point-set queries on \(G\) in \(\Oh(k+\log n)\) time, where \(k\) is the number of reported farthest points. This data structure has a construction time and size of \(\Oh(m^3)\).
%\end{theorem}
%
\begin{theorem} \label{thm::naive_brute_force_farthet_points_data_structure}
    Given a network \(G\) with \(n\) vertices and \(m\) edges. There is a data structure with size and construction time \(\Oh(m^3)\) supporting farthest-point-set queries on \(G\) in \(\Oh(k+\log n)\) time, where \(k\) is the number of reported farthest points. %This data structure has a construction time and size of \(\Oh(m^3)\).
\end{theorem}

\section{A Data Structure for Eccentricity, \texorpdfstring{$R$}{R}-Far, and Farthest-Point-Set Queries} \label{sec::far_point_queries}

% intro: kind of query, farthest-point-set queries are R-queries with R eccentricity 
When analyzing the location of a service facility in a network of roads, we may want to determine the part of the network that is ill-served, i.e., farther away from the facility than some critical threshold \(R\). We introduce a data structure for \(R\)-far queries, which consists of a point \(p\) on a network \(G\), and a value \(R > 0\). We seek the part of \(G\) with network distance at least \(R\) to the query point \(p\). 

% idea: finding functions that are higher than R, for each such edge compute part on edge farther than R, it exists because farthest point farther than R, on the other hand if function closer than R, then no point on edge farther than R.

We develop a data structure for \(R\)-far queries from a fixed edge \(uv\) and then build this data structure for every edge. To answer \(R\)-far queries from edge \(uv\) we perform two tasks: The first task is to compute the part of an edge \(st\) consisting of the \(R\)-far points from query point \(p\), i.e., the points \(q\) on \(st\) with \(d(p,q) \ge R\). The second task is to identify those edges of the network that contain \(R\)-far points without inspecting all edges. 

Consider a point \(p(\lambda) \) on edge \(uv\) with \(p(\lambda) =(1- \lambda) u + \lambda v\) for some \(\lambda \in [0,1]\), and consider another edge \(st\). Let \(\bar p(\lambda)\) be the farthest point from \(p(\lambda)\) on edge \(st\).  The set of \(R\)-far points on \(st\) is the sub-edge of points \(q\) with \(w_{q\bar p(\lambda)} \le \phi_{uv}^{st}(\lambda) - R\), since the distance to \(p\) decreases from \(\bar p\) to \(v\) and from \(\bar p\) to \(u\). 

% idea: how to find functions that are higher than R: split up segments in three groups: ascending, constant and descending, segments can be ordered within groups, we will explain next how to use segment trees to store and access the segments or rathter their projection onto x-axis descendingly in this order. 

We rephrase the task to find all edges containing \(R\)-far points as ray shooting problem. An edge \(st\) contains \(R\)-far points from \(p(\lambda)\) if and only if \(\bar p(\lambda)\) is \(R\)-far, i.e., if and only if \(d(p(\lambda), \bar p(\lambda)) = \phi_{uv}^{st}(\lambda) \ge R\). Thus, we seek the functions \(\phi_{uv}^{st}\), for all edges \(st\), whose height at \(\lambda\) is at least \(R\), or, in other words, who are intersected by a vertical ray \(\vec r\) that shoots upwards from the point \((\lambda,R)\), as shown in \Cref{fig::example_brute_force_fps::r_far}. We solve this ray shooting problem separately for the line segments of the functions \(\phi_{uv}^{st}\) with a common slope (\(w_{uv}\), zero, or \(-w_{uv}\)). We use segment trees \cite{bentley1977solutions,deberg2008more} and exploit that line segments with a common slope have a vertical order.
\begin{figure}[ht]
	\centering
	\begin{subfigure}{0.59\linewidth}
		\begin{subfigure}{\linewidth}
		\centering
		\includegraphics[page=1,scale=0.8]{example_brute_force}
		\caption{A network.} \label{fig::example_brute_force::input::again::andagain}
		\end{subfigure}
		
		\begin{subfigure}{\linewidth}
		\centering
		\includegraphics[page=7,scale=0.8]{example_brute_force}
		\caption{The \(R\)-far portion of the network.} \label{fig::example_brute_force::r_far::uv}
	\end{subfigure}%
	\end{subfigure}
	\begin{subfigure}{0.39\linewidth}
	\centering
		\includegraphics[]{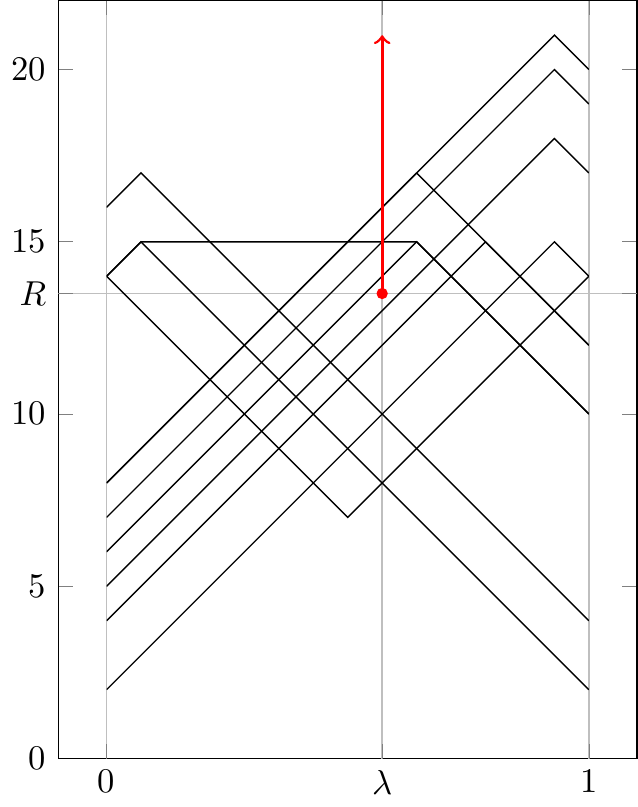}
		\caption{The functions \(\phi_{uv}^{st}\), for all edges \(st\), and an upwards shooting ray with apex \((\lambda,R)\).} \label{fig::plot::example_brute_force_fps::r_far_ray}
	\end{subfigure}
	
	\caption{An example for the correspondence between an \(R\)-far query in the network and a vertical ray stabbing query. On the upper left \subref{fig::example_brute_force::input::again::andagain}, we see a network \(G\). On the lower left \subref{fig::example_brute_force::r_far::uv}, we see the answer (red) to an \(R\)-far query at \(p(\lambda) = (1-\lambda) u + \lambda v\) on edge \(uv\) of \(G\) with \(\lambda = \frac{8}{14}\) and \(R = 13.5\). On the right \subref{fig::plot::example_brute_force_fps::r_far_ray}, we see the corresponding ray stabbing query. The functions \(\phi_{uv}^{st}\) that are stabbed by the ray (red) contain \(R\)-far points from the query point \(p(\lambda)\). \label{fig::example_brute_force_fps::r_far}} 
\end{figure}

\begin{lemma}[\cite{bentley1977solutions,deberg2008more}]
    Let \(L\) be a set of \(N\) vertically ordered line segments in \(\R^2\). There is a data structure that reports all \(k\) line segments in \(L\) that are intersected by a vertical ray shooting upwards from a query point \((x,y)\) in \(\Oh(k + \log N)\) time. This data structure has construction time and size \(\Oh(N \log N)\).
\end{lemma}

Our data structure for \(R\)-far queries on an edge \(uv\) consists of three segment trees \(T_\text{asc}\), \(T_\text{con}\), and \(T_\text{des}\) containing the ascending, constant, and descending line segments of the functions \(\phi_{uv}^{st}\) for all edges \(st\) in decreasing vertical order. These segment trees allow us to compute the \(k\) edges of the network containing \(R\)-far points from a query point \(p(\lambda) = (1- \lambda) u + \lambda v\) in \(\Oh(k + \log n)\) time with a ray shooting query for the ray shooting upwards from \((\lambda,R)\). 

\begin{lemma} \label{thm::data_structure::R_far_queries}
    Let \(uv\) be an edge in a network with \(n\) vertices and \(m\) edges. There is a data structure that supports \(R\)-far queries from points on \(uv\) in \(\Oh(k + \log n)\) time, where \(k\) is the number of edges containing \(R\)-far points from the query point. This data structure has size and construction time \(\Oh(m \log n)\). 
\end{lemma}

We can use this data structure for eccentricity queries and farthest-point set queries on edge \(uv\), as well. For eccentricity queries, we keep track of the maximum heights at \(\lambda\) of the line segments stored at the heads of the lists encountered as we follow the search path for \(\lambda\). The maximum height is the eccentricity of \(p(\lambda)\), since it is the greatest value among \(\phi_{uv}^{st}(\lambda)\) for all edges \(st\) of \(G\). This query takes \(\Oh(\log n)\) time, since we only have to inspect the heads of the \(\Oh(\log n)\) lists along the search path. For farthest-point set queries, we first determine the eccentricity of \(p(\lambda)\) and then perform a \(R\)-far query with \(R = \ecc(p(\lambda))\). We obtain our final data structure for all three types of queries by building the data structure from \cref{thm::data_structure::R_far_queries} for each edge of the network.

\begin{theorem}
Given a network \(G\) with \(n\) vertices and \(m\) edges.  There is a data structure with  size and construction time \(\Oh(m^2 \log n)\) supporting eccentricity queries, \(R\)-far queries, and farthest-point-set queries from any query point \(p\) on \(G\). Let \(k\) denote the number of edges containing \(R\)-far points from \(p\), and let \(k'\) denote the number of farthest points from \(p\) in \(G\). Using the data structure, an eccentricity query takes \(\Oh(\log n)\) time, an \(R\)-far query takes \(\Oh(k + \log n)\) time,  and a farthest-point-set query takes \(\Oh(k' + \log n)\) time. 
\end{theorem}

\section{The Minimum Eccentricity Feed-Link Problem} \label{sec::feed-links}

In this section, we solve the feed-link problem.% from the introduction. First, we repeat the definition of this problem alongside with an illustration in \cref{fig::maltese_cross::minimum_feed_link_problem_again}.

\mineccfeedlinkprobdef*
%\begin{figure}[ht]
%	\centering
%	\begin{subfigure}[t]{.47\linewidth}
%		\centering
%		\includegraphics[scale=1, page=3]{maltese_cross_example}
%		\caption{A network \(G\) with a point \(p \notin G\).}\label{fig::maltese_cross::feedlink_example_input_again}
%	\end{subfigure}\hspace{0.02\linewidth}
%	\begin{subfigure}[t]{.47\linewidth}
%		\centering
%		\includegraphics[scale=1, page=5]{maltese_cross_example}
%		\caption{An optimal feed-link.}\label{fig::maltese_cross::optimal_feed_link_again}
%	\end{subfigure} 
%	\caption{An instance of the minimum eccentricity feed-link problem \subref{fig::maltese_cross::feedlink_example_input} with the network from  		\cref{fig::maltese_cross::example_eccentricity}. The anchor \(q\) of an optimal feed-link (dashed) \subref{fig::maltese_cross::optimal_feed_link} balances the distance to the eccentric points \(\bar q\) and \(\bar q'\).}\label{fig::maltese_cross::minimum_feed_link_problem_again}
%\end{figure} 

We treat two versions of the minimum eccentricity feed-link problem. In the static version, we have a fixed network \(G\) and a fixed point \(p\) that we wish to connect to \(G\). Here, we seek the optimal feed-link for \(G\) and \(p\). In the query version of the problem, we have a fixed network \(G\) and a query consists of a point \(p\) that we wish to connect to \(G\). Here, we seek a data structure that can answer queries of this type efficiently. 

\subsection{The Static Version}

The dependence upon the eccentricity \(\ecc_G(q)\) of the anchor point \(q\) in \eqref{eq::target_function_minimum_eccentricity_feed_link_problem} shows the  connection between farthest-point information on networks and the minimal eccentricity feed-link problem. This dependence determines necessary conditions on the optimal feed-link. 

\begin{lemma} \label{thm::optimal_anchor_is_local_minimum_of_eccentricity}
    Let \(G\) be a geometric network. Let the point \(q\) on \(G\) be the anchor of an optimal feed-link for the point \(p \in \R^2\). Then the eccentricity on \(G\) has a local minimum at \(q\). Furthermore, if \(q\) is located on an edge \(uv\) of \(\ED(G)\) with constant eccentricity, then \(q\) is the closest point to \(p\) on \(uv\) with respect to Euclidean distance. 
\end{lemma}
\begin{proof} Let \(uv\) be the edge of \(\ED(G)\) containing the optimal anchor \(q\) of a feed-link from \(p\) to \(G\).
\item Case (1): Let the eccentricity be increasing on \(uv\) with \(\ecc_G(u) < \ecc_G(v)\). Then we have%
	\begin{align*}
		\ecc_{G+pq}(p) &= \abs{pq} + \ecc_G(q) \\
		&\ge \abs{pu} - \abs{uq} + \ecc_G(q) &\\
		&= \abs{pu} - \abs{uq} + \ecc_G(u) + \abs{uq} &{\text{ as the eccentricity is increasing on } uv,}  \\
		&= \abs{pu} + \ecc_G(u) &\\
		&= \ecc_{G+pu}(p) \enspace .
	\end{align*}
	Therefore, the optimal anchor among all points on \(uv\) is \(u\), i.e., \(q = u\). 
	
\item Case (2): Let the eccentricity be constant on \(uv\), and let \(q'\) be the closest point from \(p\) on \(uv\).%
\begin{align*}
	\ecc_{G+pq}(p) &= \abs{pq} + \ecc_G(q) &\\
	&= \abs{pq} + \ecc_G(q') &{\text{as the eccentricity is constant on } uv,} \\
	&\ge \abs{pq'} + \ecc_G(q') &\text{ by choice of } q', \\
	& = \ecc_{G+pq'}(p) \enspace .
\end{align*}
Therefore, the optimal anchor among all points on \(uv\) is \(q'\), i.e., \(q = q'\).
\end{proof}

Using \cref{thm::optimal_anchor_is_local_minimum_of_eccentricity}, we can solve the minimum eccentricity feed-link problem as follows. First, we compute the eccentricity diagram \(\ED(G)\) of the network \(G\). Second, we read the sub-edges of \(G\) with locally minimal eccentricity from \(\ED(G)\). These sub-edges are the edges of \(\ED(G)\) with constant eccentricity and the vertices of \(\ED(G)\)---which we treat as sub-edges reduced to single points---whose neighbours in \(\ED(G)\) have greater eccentricity. Let \(S\) be the set of sub-edges with locally minimal eccentricity. Third, we determine a candidate \(q\) for the optimal anchor on each sub-edge with locally minimal eccentricity. Among these candidates we select the one with the lowest value of \(\abs{pq} + \ecc_G(q)\). An example is show in \cref{fig::sunlet::minimum_eccentricity_feed_link::example}. The first step takes \(\Oh(m^2 \log n)\) time as discussed in \cref{sec::eccentricity_diagrams}. The second step can be done alongside with the computation of the eccentricity diagram. The third step takes \(\Oh(\abs{S})\) time, where \( \abs{S} \in  \Oh(m^2)\). 
\begin{figure}[!ht]
	\centering
	\begin{subfigure}{0.32\linewidth}
		\centering
		\includegraphics[scale=0.6, page=6]{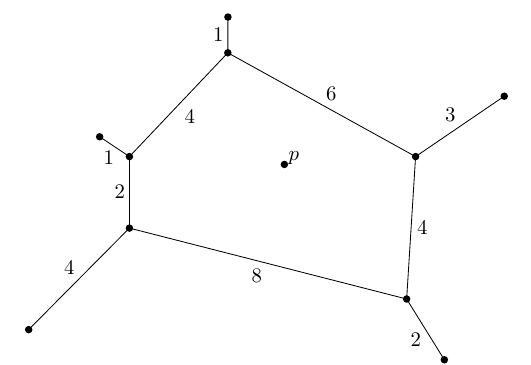}
		\subcaption{A network \(G\) with a site \(p\).\label{fig::sunlet::minimum_eccentricity_feed_link::instance}}
	\end{subfigure}%
	\begin{subfigure}{0.32\linewidth}
		\centering
		\includegraphics[scale=0.6, page=7]{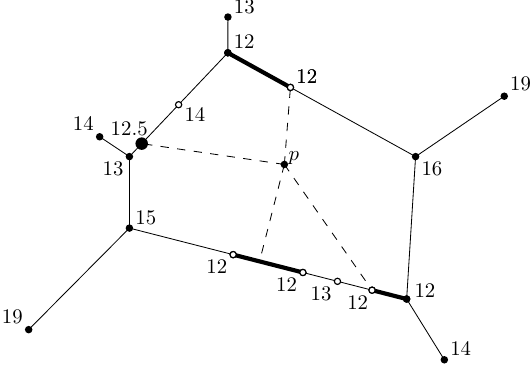}
		\subcaption{Canditate feed-links.\label{fig::sunlet::minimum_eccentricity_feed_link::candidates}}
	\end{subfigure}
%	
%	\begin{subfigure}{0.32\linewidth}
%		\centering
%		\includegraphics[scale=0.6, page=8]{sunlet}
%		\subcaption{The optimal feed-link. \label{fig::sunlet::minimum_eccentricity_feed_link::optimal_feed_link}}
%	\end{subfigure}%
	\begin{subfigure}{0.32\linewidth}
		\centering
		\includegraphics[scale=0.6, page=9]{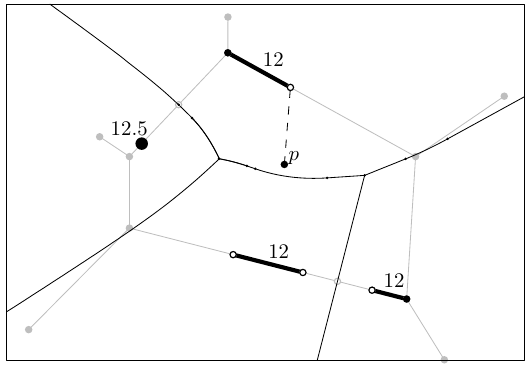}
		\subcaption{The optimal feed-link. \label{fig::sunlet::minimum_eccentricity_feed_link::apollonius_diagram}}
	\end{subfigure}%
	\caption{From left \subref{fig::sunlet::minimum_eccentricity_feed_link::instance} to right \subref{fig::sunlet::minimum_eccentricity_feed_link::apollonius_diagram}: An instance of the minimum eccentricity feed-link problem. The candidates for optimal feed-links (dashed) and local minima (thick, black) of the eccentricity function.
%	The best feed-link \(pq\) together with 
The subdivision of the plane into regions with a common sub-edge containing an optimal feed-link.\label{fig::sunlet::minimum_eccentricity_feed_link::example}}
\end{figure} 

\begin{theorem} 
	Given a geometric network \(G\) with \(n\) vertices and \(m\) edges, and a point \(p \in \R^2\). Assume we are given the \(\ell\) sub-edges of \(G\) with locally minimal eccentricity. Then we can solve the minimum eccentricity feed-link problem with respect to \(G\) and \(p\) in \(\Oh(\ell)\) time.
\end{theorem}

\begin{corollary} \label{thm::minimum_eccentricity_feed_link_problem_static_solution}
    Given a geometric network \(G\) with \(n\) vertices and \(m\) edges, and a point \(p \in \R^2\). We can solve the minimum eccentricity feed-link problem with respect to \(G\) and \(p\) in \(\Oh(m^2 \log n)\) time.
\end{corollary}

\subsection{The Query Version}

Now we address the query version, where the network \(G\) is fixed and a query consists of the point \(p\). Throughout the following let \(S\) be the set of sub-edges of \(G\) with locally minimal eccentricity, and let \(\ell = \abs{S}\). 

Using the solution for the static problem, we can create a data structure with size \(\Oh(\ell)\) and construction time \(\Oh(m^2 \log n)\) that answers queries for the optimal feed-link in \(\Oh(\ell)\) time. This data structure consists of the set \(S\), and we obtain it by computing the eccentricity diagram and recording the local minima of the eccentricity. We improve the query time---at the expense of space consumption and construction time---by rephrasing the minimum eccentricity feed-link problem as a point location problem in a special type of Voronoi diagram.

We denote the Euclidean distance between a point \(p \in \R^2\) and a segment \(s \in S\) by \(d_2(p,s)\), i.e., \(d_2(p,s) = \min_{q \in s} \abs{pq}\). By definition, the eccentricity with respect to the network \(G\) is constant on all segments \(s \in S\). We write \(\ecc(s)\) to denote the eccentricity of the points on \(s\). With this notation, the optimal feed-link is the \emph{closest} sub-edge \(s \in S\) with respect to the additively weighted Euclidean distance \(d_2(p,s) + \ecc(s)\). 

Conversely, consider the Voronoi diagram of the line segments in \(S\) with respect to the additively weighted Euclidean distance where the weight of a segment \(s \in S\) is its eccentricity \(\ecc(s)\). This diagram splits the plane into regions whose points have a common closest segment in \(S\) with respect to the additively weighted Euclidean distance. In other words, the points in each region have their anchor of an optimal feed-link on a common sub-edge in \(S\). \Cref{fig::sunlet::minimum_eccentricity_feed_link::apollonius_diagram} shows an example of this kind of Voronoi diagram. 

\begin{definition}
	Let \(S\) be a set of line segments in the plane with weights \(w_s \in \R\) for each \(s \in S\). 
	\begin{enumerate}[(i)]
		\item The \emph{additively weighted distance} of a point \(p \in \R^2\) and a line segment \(s \in S\) is the Euclidean distance \(d_2(p,s)\) of \(p\) and \(s\) plus the weight \(w_s\) of the line segment \(s\).
		\item We call the set of points \(p \in \R^2\) to whom a line segment \( s \in S\) has the lowest additively weighted distance among all line segments in \(S\) the \emph{additive weight Voronoi cell} of \(s\), and denote it by \(V_+(s)\), i.e.,
		\[
			V_+(s) \coloneqq \set*{p \in \R^2 \colon \forall s' \in S \colon d_2(p,s) + w_s \le d_2(p,s') + w_{s'} }.
		\] 
		\item We call the subdivision of the plane into the set \(\bigcup_{s\in S} \partial V_+(s)\) and the connected regions of 
		\[\R^2 \setminus \left(\bigcup_{s\in S} \partial V_+(s)\right),\]  the \emph{additively weighted Voronoi diagram} of the line segments in \(S\) with respect to the weights \(w_s\), \(s \in S\). %, and denote it by \(\mathcal{VD_+}(S)\). Further, we call \(\bigcup_{s\in S} \partial V_+(s)\) the \emph{skeleton} of \(\mathcal{VD_+}(S)\).
	\end{enumerate}
\end{definition}

Voronoi diagrams of points with additive weights \cite[Section~3.1.2]{oktabe2000spatial} and Voronoi diagrams of line segments \cite[Section~3.5]{oktabe2000spatial} have received considerable attention in the literature and are thus well studied concepts. To the best of the authors' knowledge, the following quotation is the only direct mentioning: 
\begin{quote}
\enquote{In general, the Voronoi diagram of segments where each segment carries an additive weight is not a well-behaved Voronoi diagram: Voronoi regions can be disconnected, and the diagram can have quadratic complexity.}---\textcite{cheong2006throwing}
\end{quote}%

The comprehensive study of additively weighted Voronoi diagrams of line segments is beyond the scope of this work. Nonetheless, we summarize a few observations about this type of Voronoi diagram. \Cref{fig::bisector_of_weighted_segments::bisector_with_examples} shows the (ill-behaved) additively weighted Voronoi diagram of two line segments. 
\begin{figure}
    \centering
    \begin{subfigure}[t]{0.48\linewidth}
        \centering
        \includegraphics[scale=0.9,page=3]{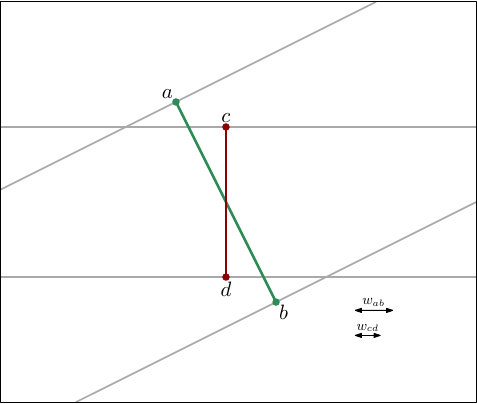}
        \caption{Two line segments and their weights.\label{fig::bisector_of_weighted_segments::segments_and_weights}}
    \end{subfigure} 
    \begin{subfigure}[t]{0.48\linewidth}
        \centering
        \includegraphics[scale=0.9,page=4]{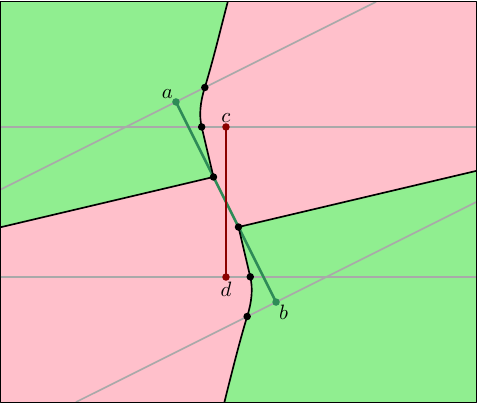}
        \caption{A part of the subdivision of the plane into regions (green) with smaller additively weighted Euclidean distance to \(ab\) than to \(cd\) and their complement (red). \label{fig::bisector_of_weighted_segments::bisector}}
    \end{subfigure}
    
%    \begin{subfigure}{0.48\linewidth}
%        \centering
%        \includegraphics[scale=1,page=6]{bisector_of_weighted_segments}
%        \caption{A point \(p\) with \(d_w(p,ab) \le d_w(p,cd)\).\label{fig::bisector_of_weighted_segments::example_ab}}
%    \end{subfigure} \quad 
%    \begin{subfigure}{0.48\linewidth}
%        \centering
%        \includegraphics[scale=1,page=5]{bisector_of_weighted_segments}
%        \caption{A point \(p\) with \(d_w(p,cd) \le d_w(p,ab)\). \label{fig::bisector_of_weighted_segments::example_cd}}
%    \end{subfigure}
    \caption{The bisector \subref{fig::bisector_of_weighted_segments::bisector} of two line segments with additive weights \subref{fig::bisector_of_weighted_segments::segments_and_weights}. The visible part of this bisector consists of line segments, parabolic arcs, and hyperbolic arcs. The region of points closer to the green segment are non-convex and disconnected; this misbehaving bisector splits the plane in three. \label{fig::bisector_of_weighted_segments::bisector_with_examples}}
\end{figure}

\begin{theorem}
The additively weighted Voronoi diagram of \(\ell\) line segments is an planar subdivision of size \(\Theta(\ell^2)\) whose edges are parts of lines (lines, rays, line segments), parts of parabolas (parabolas, parabolic rays, parabolic arcs), and parts of hyperbolas (hyperbolas, hyperbolic rays, hyperbolic arcs). 
\end{theorem}

Using common techniques from planar point location \cite{snoeyink2004point,edelsbrunner1986optimal}, we determine the region that contains a query point \(p\) in \(\Oh(\log \ell)\) time. The storage requirement and construction time of this data structure are linear in the size of the planar subdivision, provided that the subdivision is monotone. We can make the subdivision monotone in \(\Oh(\ell^2 \log \ell)\) time using plane sweep. 

We briefly discuss the construction of additively weighted Voronoi diagrams of line segments. The algorithms for abstract Voronoi diagrams \cite{klein2009abstract} only work for (generalizations of) Voronoi diagrams whose Voronoi cells are connected, a property that is violated in our case. Despite the lack of existing theory, we can compute the desired diagram using the relationship between Voronoi diagrams and lower envelopes in three dimensions via lifting maps. \textcite{setter2010constructing} provide a comprehensive review of this observation by \textcite{edelsbrunner1986voronoi} and its practical implications. In a nutshell the idea is as follows: consider the lower envelope of the graphs of the distance functions of all sites of the Voronoi diagram. The projection of this envelope onto the plane yields the desired Voronoi diagram.  \textcite{agarwal1996overlay} provide a divide-and-conquer algorithm that computes the lower envelope of the (weighted) distance functions of \(\ell\) sites in \(\Oh(\ell^{2+\epsilon})\) time. The randomized version of this algorithm, which was proposed by \textcite{setter2010constructing}, accomplishes the same task in \(\Oh(\ell^2\log \ell)\) expected time. 

\begin{theorem} \label{thm::data_structure_minimum_eccentricity_feed_link_query}
 Let \(G\) be a geometric network with \(n\) vertices and \(m\) edges. Furthermore, let \(\ell\) be the number of sub-edges of \(G\) with locally minimal eccentricity. There is a data structure that can perform queries for a minimum eccentricity feed-link for any point \(p \in \R^2\) in \(\Oh(\log \ell)\) time. This data structure has a space requirement of \(\Oh(\ell^{2})\). It can be constructed in \(\Oh(\ell^{2+\epsilon})\) time or, alternatively, in   \(\Oh(\ell^2 \log \ell)\) expected time, both provided that the eccentricity diagram of \(G\) is known a-priori. 
\end{theorem}

\section{Conclusions and Future Work} \label{sec::future_work}

We introduced new notions to capture farthest-point information in networks as well as data structures to store and access this information efficiently. We seek to improve the bounds on the construction time and space requirements in future work. For instance, our approach ignores any structure that the network might have and requires all pairs shortest path distances.% We will discuss cactus networks that admit a divide-and-conquer strategy \cite{grimm2012charting} and study planar networks \cite{grimm2012charting} as well as geometric spanners networks. 

We presented the feed-link problem that kindled this research alongside with a first solution for its static and query version. The feed-link problem can be extended in many ways. For instance, we could connect several sites simultaneously to a network minimizing the largest distance to the nearest site. Further, we could require the extension of a planar network to be planar as well or add other restrictions such as obstacles. \textcite{aronov2011connect} discuss the latter two for the minimum dilation feed-link problem. 

Finally, the additively weighted Voronoi diagram of line segments demands more investigation, because of its relation to the query version of the feed-link problem. 

\printbibliography

\end{document}